\documentclass[preprint,10pt]{elsarticle}

\usepackage{amssymb}

\usepackage{datetime}

\usepackage{hyperref}

\usepackage{graphicx}

\usepackage{amsmath}
\usepackage{amssymb}
\usepackage{amsfonts}
\usepackage{mathptmx} 
\usepackage{datetime}

\newtheorem{lem}{Lemma}
\newtheorem{thm}{Theorem}

\newcommand{\tanc}{\mathrm{tanc}}

\DeclareMathAlphabet{\bit}{OML}{cmm}{b}{it}
\def\<{\leqslant}           
\def\>{\geqslant}           

\def\d{\partial}
\def\wh{\widehat}
\def\wt{\widetilde}

\def\Re{\mathrm{Re}}   
\def\Im{\mathrm{Im}}   
\def\rprod{\mathop{\overrightarrow{\prod}}}
\def\lprod{\mathop{\overleftarrow{\prod}}}

\def\cH{\mathcal{H}}   
\def\mA{\mathbb{A}}    
\def\mR{{\mathbb R}}    
\def\mC{\mathbb{C}}    

\def\Tr{\mathrm{Tr}}       
\def\rT{{\rm T}}        
\def\bE{\mathbf{E}}    
\def\bM{\mathbf{M}}    

\def\re{{\rm e}}        
\def\rd{{\rm d}}        

\def\cL{\mathcal{L}}

\def\bJ{\mathbf{J}}

\def\x{\times}
\def\ox{\otimes}

\def\fF{\mathfrak{F}}

\def\fH{\mathfrak{H}}

\def\mP{\mathbb{P}}

\def\cR{\mathcal{R}}

\def\cG{\mathcal{G}}

\def\cP{\mathcal{P}}

\def\cA{\mathcal{ A}}
\def\cB{\mathcal{ B}}

\def\mH{\mathbb{H}}
\def\mS{\mathbb{S}}

\def\ups{\upsilon}
\def\Ups{\Upsilon}

\def\diag{\mathop{\mathrm{diag}}}    

\journal{the Journal of the Franklin Institute}

\begin{document}

\begin{frontmatter}



\title{State-space computation of quadratic-exponential functional rates
for linear quantum stochastic  systems}


\author{Igor G. Vladimirov (corresponding author)}
\ead{igor.g.vladimirov@gmail.com}
\author{Ian R. Petersen}
\ead{i.r.petersen@gmail.com}
\address{School of Engineering, Australian National University, Acton, 
Canberra, 2601, ACT, Australia}


\begin{abstract}
This paper is concerned with infinite-horizon growth rates of quadratic-exponential functionals (QEFs) for linear quantum stochastic systems driven by multichannel bosonic fields. Such risk-sensitive performance criteria  impose an exponential penalty on the integral of a quadratic function of the system variables,
and their minimization improves robustness properties of the system with respect to quantum statistical uncertainties and makes its behaviour more conservative in terms of tail distributions. We use a frequency-domain representation  of the QEF growth rate for the invariant Gaussian quantum state of the
system with vacuum input fields in order to compute it in state space. The QEF rate is related to a similar functional
for a classical stationary Gaussian random process generated by 
an  infinite cascade of 
linear systems. A truncation of this shaping filter allows the QEF rate to be computed with any accuracy by solving a recurrent sequence of algebraic Lyapunov equations together with an algebraic  Riccati equation. The state-space computation of the QEF rate and its comparison with the frequency-domain results are demonstrated by a numerical example for an open quantum harmonic oscillator.

\end{abstract}



\end{frontmatter}



\section*{Abbreviations}

\begin{tabular}{ll}
ALE & algebraic Lyapunov equation\\

ARE & algebraic Riccati equation\\

CCR & canonical commutation relation\\

LCTI & linear continuous time invariant \\

LQG & linear-quadratic-Gaussian\\

ODE & ordinary differential equation\\

OQHO & open quantum harmonic oscillator \\

PDE & partial differential equation\\

QEF & quadratic-exponential functional \\

QSDE & quantum stochastic differential equation\\

SDE & stochastic differential equation
\end{tabular}

\section{Introduction}
\label{sec:intro}

Linear quantum stochastic systems, or open quantum harmonic oscillators (OQHOs),  provide an important class of tractable models of open quantum dynamics which is concerned with the interaction of quantum mechanical  systems with their environment. The latter may include other quantum or classical systems (for example, measuring devices) and    quantum fields. In the framework of the Hudson-Parthasarathy calculus \cite{HP_1984,P_1992,P_2015},   these models are equipped with noncommuting continuous dynamic variables (such as the quantum mechanical positions and momenta \cite{S_1994})  whose time evolution is governed by linear quantum stochastic differential equations (QSDEs)  driven by quantum analogues of the standard Wiener process \cite{KS_1991}. The OQHOs are employed as building  blocks  in linear quantum systems theory \cite{NY_2017,P_2017} which develops methods for performance analysis and synthesis of such systems with certain dynamic properties. These developments give rise to control and filtering settings for interconnections of systems \cite{GJ_2009,JG_2008},  consisting, for example,  of a quantum plant and a quantum or classical feedback controller or observer with direct or field-mediated coupling \cite{ZJ_2012}.

By analogy with linear-quadratic-Gaussian  (LQG) control for classical linear sto\-chas\-tic systems  \cite{AM_1971}, the performance of quantum networks is often described using mean square optimality criteria in the form of quadratic cost functionals to be minimised \cite{EB_2005,NJP_2009,WM_2010}. In particular,
as in the Kolmogorov-Wiener-Hopf-Kalman filtering theory,  the quadratic criteria serve to quantify and improve the quality of observers in quantum filtering problems in terms of the mean square value of the estimation error, that is,  the  discrepancy between the quantum plant variables and their estimates \cite{MJ_2012}. Similarly to their  classical predecessors,  the mean square optimality criteria for linear quantum stochastic systems are a limiting case of  appropriate quantum mechanical counterparts \cite{B_1996,J_2004,J_2005,VPJ_2018a} of quadratic-exponential cost functionals,  which originate  from classical risk-sensitive control \cite{BV_1985,J_1973,W_1981,W_1990}. One of these extensions  underlies the original quantum risk-sensitive control formulation  \cite{J_2004,J_2005} and employs time-ordered  exponentials, which differ from (though have links \cite{VPJ_2019a} with) the usual ones because of the noncommutativity of the  system variables.

The general structure of  the classical risk-sensitive performance criteria is retained by
the quadratic-exponential functional \cite{VPJ_2018a} (QEF) which is the averaged exponential of the integral of a  quadratic function of the quantum system variables over a finite time horizon. In contrast to the mean square criteria, the QEF is organised as a higher-order mixed moment of the quantum variables at different times. Due to this specific structure, the QEF  yields upper bounds \cite{VPJ_2018b}  on the worst-case mean square costs in the presence of quantum statistical uncertainty described in terms of quantum relative entropy \cite{OP_1993,OW_2010,YB_2009} of the actual system-field state with respect to its nominal model. This role of the QEF  is similar to the connections between the classical risk-sensitive criteria and minimax LQG control \cite{DJP_2000,P_2006,PJD_2000}. The QEF also provides exponential upper bounds on the tail distributions for quadratic functions of the quantum system trajectories \cite{VPJ_2018a}, which corresponds to the large deviations theory for classical  random processes \cite{DE_1997,V_2008}. These bounds depend on the QEF monotonically, so that its minimisation secures a more conservative and robust dynamic behaviour of the open quantum system. The robustness properties, including reduced sensitivity to unmodelled dynamics (such as nonlinearities), and  controlled isolation  of the quantum  system from its surroundings   are relevant, for example,    to applications in quantum optics and quantum information processing \cite{DP_2010,NC_2000,WM_2008}.

Since, in  the noncommutative quantum case,   the QEF differs both from the classical predecessors and its time-ordered exponential counterpart \cite{J_2004,J_2005},  the performance analysis and optimal control synthesis with QEF criteria demand methods for computing and  minimising such functionals. In addition to their primary relevance to quantum risk-sensitive control, such methods are also of interest on a broader scale of quantum probabilistic and algebraic connections with the moment-generating and partition functions for quadratic Hamiltonians in quantum statistical mechanics \cite{BB_2010,PS_2015,S_1994},
the operator exponential structures in the context
of operator algebras \cite{AB_2018},  and quantum mechanical extensions of the  L\'{e}vy area \cite{CH_2013,H_2018}.

Resulting from recent publications on a parametric randomization technique \cite{VPJ_2018c}, quantum Karhunen-Loeve expansions \cite{VPJ_2019b,VJP_2019} and Girsanov type representations \cite{VPJ_2020a}  for computing the QEF over a bounded time interval, a frequency-domain formula has been established in \cite{VPJ_2020_IFAC} for the infinite-horizon  asymptotic growth rate of the logarithm of the QEF for invariant  Gaussian states of stable OQHOs with vacuum input  bosonic fields \cite{P_1992}. This relation, which  has subsequently been extended to more general stationary Gaussian quantum processes \cite{VPJ_2021},  expresses the  QEF growth rate in terms of the Fourier transforms of the real and imaginary parts of the invariant quantum covariance function of the system variables. These matrix-valued spectral functions enter the QEF rate through their compositions with trigonometric functions and the log-determinant,  thus destroying the meromorphic structures which play a role in tractability of the $\cH_\infty$-entropy integral \cite{AK_1981,MG_1990} for the classical QEF rate. Nevertheless, the quantum QEF rate lends itself to numerical computation in the frequency domain using a homotopy algorithm \cite{VPJ_2020_IFAC,VPJ_2021}  similar to that for solving parameter dependent algebraic equations \cite{MB_1985}. However, each step  of this algorithm involves a time-consuming  high-resolution numerical integration over the frequency range, which decreases its practicality. At the same time, these  frequency-domain methods have already found a preliminary application to optimality conditions for measurement-based feedback control with QEF criteria \cite{VJP_2020_IFAC}.

The present paper builds on the frequency-domain representation of the QEF rate in \cite{VPJ_2020_IFAC,VPJ_2021}
for the invariant Gaussian state of  an OQHO,  driven by vacuum quantum fields,  and develops  an  approach to  its computation in state space. For this purpose, we relate the QEF rate through a nonrational spectral density to a similar functional
for an auxiliary  classical stationary Gaussian random process. The latter is produced from a standard Wiener process by
an  infinite cascade of
linear systems whose state-space matrices are computed through a recurrent sequence of algebraic Lyapunov equations (ALEs).  The complicated infinite dimensional  structure of this shaping filter comes from the trigonometric functions in the frequency-domain representation of the QEF rate  mentioned above. This key element of the QEF rate computation in state space is based on a novel spectral factorisation using a ``system transposition'' technique for rearranging mixed products of classical linear systems with their duals, which  resembles the Wick ordering \cite{W_1950} for quantum mechanical annihilation and creation operators.
A truncation of this shaping filter, combined with an additional algebraic Riccati equation (ARE),   allows the QEF rate to be computed with any accuracy.  The state-space computation of the QEF rate with different truncation orders and its comparison with the frequency-domain results are demonstrated by a numerical example for a two-mode  OQHO.

The paper is organised as follows.
Section~\ref{sec:sys} describes the class of linear quantum stochastic systems being considered and their invariant Gaussian states.
Section~\ref{sec:QEF} specifies the QEF for such a system and revisits the frequency-domain formula for its growth rate.
Section~\ref{sec:Delta} provides a spectral density representation for  the QEF growth rate.
Section~\ref{sec:inf} develops an infinite cascade factorisation of this spectral density involving a recurrent sequence of ALEs.
Section~\ref{sec:class} relates the quantum QEF rate to a similar functional for a classical Gaussian process with an infinite-dimensional shaping filter.
Section~\ref{sec:trunc} computes the classical QEF rate using a truncation of the filter together with an ARE.
Section~\ref{sec:polrat} discusses a square root polynomial 
approximation for the entire function associated with the infinite cascade.
Section~\ref{sec:exp} provides an illustrative numerical example of state-space computation of the QEF rate  for an OQHO.
Section~\ref{sec:conc} makes concluding remarks and outlines further directions of research.

\section{Open quantum harmonic oscillators}
\label{sec:sys}

We consider an open quantum harmonic oscillator (OQHO) with an even number $n$ of dynamic variables $X_1, \ldots, X_n$, which are time-varying self-adjoint  operators on a complex separable Hilbert space $\fH$, satisfying at every moment of time the canonical commutation relations (CCRs)
\begin{equation}
\label{XCCR}
    [X,X^\rT] = 2i\Theta,
    \qquad
    X := \begin{bmatrix}
      X_1\\
      \vdots\\
      X_n
    \end{bmatrix}
\end{equation}
with a nonsingular matrix $\Theta = -\Theta^\rT \in \mR^{n\x n}$. Here, the matrix transpose $(\cdot)^\rT$ is applied to vectors of operators as if they consisted of scalars  (vectors are organised as columns unless indicated otherwise), and the commutator $[\alpha, \beta]:= \alpha \beta - \beta \alpha$ of linear operators is extended to vectors $\xi:= (\xi_j)_{1\< j \< a}$ and $\eta:= (\eta_k)_{1\< k \< b}$ of operators 
as $[\xi, \eta^\rT] := ([\xi_j, \eta_k])_{1\< j \< a, 1\< k \< b}$.
The system  also has an even number $m$ of self-adjoint output  quantum variables $Y_1, \ldots, Y_m$ on $\fH$, assembled into a vector
$$
    Y := \begin{bmatrix}
      Y_1\\
      \vdots\\
      Y_m
    \end{bmatrix},
$$
 and is governed by linear QSDEs \cite{NY_2017,P_2017}
\begin{align}
\label{dX}
    \rd X & = AX \rd t + B\rd W,\\
\label{dY}
    \rd Y & = CX \rd t + \rd W,
\end{align}
with constant coefficients comprising the matrices $A \in \mR^{n\x n}$, $B \in \mR^{n\x m}$, $C \in \mR^{m\x n}$ which are described below.
These QSDEs are driven by the vector
$$
    W := \begin{bmatrix}
      W_1\\
      \vdots\\
      W_m
    \end{bmatrix}
$$
of self-adjoint quantum Wiener processes $W_1, \ldots, W_m$ on a symmetric Fock space $\fF$ \cite{P_1992} with the Ito table
\begin{equation}
\label{Omega}
    \rd W \rd W^\rT
    = \Omega \rd t,
    \qquad
    \Omega: = I_m + iJ,
\end{equation}
where
\begin{equation}
\label{JJ}
    J
    :=
    \bJ \ox I_{m/2}
    =
    \begin{bmatrix}
      0 & I_{m/2}\\
    -I_{m/2} & 0
    \end{bmatrix}.
\end{equation}
Here, $\ox$ is the Kronecker product of matrices, $I_r$ is the identity matrix of order $r$, and the matrix
\begin{equation}
\label{bJ}
\bJ: = {\begin{bmatrix}
        0 & 1\\
        -1 & 0
    \end{bmatrix}}
\end{equation}
spans the one-dimensional subspace of antisymmetric matrices of order 2. In accordance with (\ref{Omega}),  the matrix $J$ in (\ref{JJ}) specifies the commutation structure of $W$ as
\begin{equation}
\label{WWcomm}
  [W(s), W(t)^\rT]
  =
  2i \min(s,t)J,
  \qquad
  s, t \> 0.
\end{equation}
The matrices 
\begin{equation}
\label{ABC}
    A = 2\Theta (R + M^\rT J M),
     \qquad
     B = 2\Theta M^\rT,
     \qquad
     C = 2JM
\end{equation}
in (\ref{dX}), (\ref{dY})
are parameterised
by the energy and coupling matrices $R = R^\rT \in \mR^{n\x n}$, $M \in \mR^{m\x n}$
 specifying the system Hamiltonian 
$    
    \frac{1}{2} X^\rT R X$ 
and the vector $MX$ of $m$ system-field coupling operators, which describe the energetics of the quantum system and its interaction with the external fields. For any nonsingular matrix $\sigma  \in \mR^{n\x n}$,   the transformation
\begin{equation}
\label{SX}
    X\mapsto \sigma X
\end{equation}
of the system variables leads to another OQHO with appropriately modified CCR, energy and coupling matrices:
\begin{equation}
\label{STRM}
    (\Theta, R, M)
    \mapsto
    (\sigma \Theta \sigma ^\rT,
    \sigma ^{-\rT} R \sigma ^{-1},
    M\sigma ^{-1}),
\end{equation}
where $(\cdot)^{-\rT}:= ((\cdot)^{-1})^\rT$. The state-space matrices (\ref{ABC}) are transformed  similarly to those of classical linear systems:
\begin{equation}
\label{SABC}
    (A,B,C)
    \mapsto
    (\sigma A\sigma ^{-1},
    \sigma B,
    C\sigma ^{-1}).
\end{equation}
However, their
special structure is of quantum nature and imposes the physical realizability (PR) constraints \cite{JNP_2008}
\begin{align}
\label{PR1}
    A\Theta + \Theta A^\rT + \mho & = 0, \\
\label{PR2}
    \Theta C^\rT + BJ & = 0,
\end{align}
where
\begin{equation}
\label{BJB}
    \mho:=
    BJB^\rT
    =
    -\mho^\rT
\end{equation}
is an auxiliary real matrix of order $n$,  which inherits its antisymmetry from the CCR matrix $J$ of the quantum Wiener  process $W$ in (\ref{JJ}).
The PR conditions (\ref{PR1}), (\ref{PR2}) are closely related to the preservation of the CCRs (\ref{XCCR}) together with the commutativity
$$
    [X(t), Y(s)^\rT] = 0,
    \qquad
    t\> s \> 0.
$$
In turn, this nondemolition property \cite{B_1983,B_1989} is related to the fact  that the output field $Y$ of the system has the same commutation structure as the quantum Wiener process $W$ in (\ref{WWcomm}):
$$
  [Y(s), Y(t)^\rT]
  =
  2i \min(s,t)J,
  \qquad
  s, t \> 0.
$$
Also note that the CCRs (\ref{XCCR})  for the system variables are part of their two-point CCRs \cite{VPJ_2018a}
\begin{equation}
\label{XXCCR}
    [X(s), X(t)^\rT]
    =
    2i\Lambda(s-t),
    \qquad
    s,t\>0,
\end{equation}
with
\begin{equation}
\label{Lambda}
    \Lambda(\tau)
     :=
    \left\{
    {\small\begin{matrix}
    \re^{\tau A}\Theta & {\rm if}\  \tau\> 0\\
    \Theta\re^{-\tau A^{\rT}} & {\rm if}\  \tau< 0\\
    \end{matrix}}
    \right.
    =
    -\Lambda(-\tau)^\rT,
    \qquad
    \tau \in\mR,
\end{equation}
from which (\ref{XCCR}) follows since
\begin{equation}
\label{Lambda0}
    \Lambda(0)= \Theta.
\end{equation}
The CCRs (\ref{XXCCR}) and their one-point case (\ref{XCCR}) are a consequence of the commutation structure of the system variables, as well as the external quantum fields,  and   hold regardless of a particular quantum state. The latter is described by a positive semi-definite self-adjoint density operator $\rho$ of unit trace (that is,  $\rho = \rho^\dagger \succcurlyeq 0$ and $\Tr \rho = 1$,  with $(\cdot)^\dagger$ the operator adjoint) on the system-field space $\fH:= \fH_0\ox \fF$, where $\fH_0$ is the initial system space for the action of $X_1(0), \ldots, X_n(0)$. The quantum state specifies the expectation
\begin{equation}
\label{bE}
    \bE \zeta := \Tr(\rho \zeta)
\end{equation}
for quantum variables $\zeta$ on $\fH$.
We will be concerned with the tensor-product states
\begin{equation}
\label{rho}
    \rho = \rho_0 \ox \ups,
\end{equation}
where $\rho_0$ is the initial quantum state of the system on $\fH_0$, and $\ups$ is the vacuum state for the quantum Wiener process $W$ on $\fF$ with the quasi-characteristic functional (QCF) \cite{CH_1971,HP_1984,P_1992}
$$
    \bE \re^{i\int_0^T f(t)^\rT \rd W(t)}
    =
    \re^{-\frac{1}{2} \int_0^T |f(t)|^2\rd t},
    \qquad
    f \in L^2([0,T], \mR^m),
    \quad
    T >0,
$$
where the averaging (\ref{bE}) reduces to that over the vacuum field state $\ups$,
since $\bE \eta = \Tr (\ups \eta)$ for any quantum variable on the Fock space $\fF$. In this case (of vacuum input fields, statistically independent of the initial system variables in view of (\ref{rho})), and assuming that the matrix $A$ in (\ref{ABC}) is Hurwitz,  the system variables of the OQHO have  a unique invariant multi-point Gaussian quantum state \cite{VPJ_2018a}. The latter is specified by the zero  mean $\bE X = 0$ and the two-point quantum covariance function
\begin{equation}
\label{EXX}
  \bE (X(s)X(t)^\rT) = P(s-t) + i\Lambda(s-t),
  \qquad
  s, t \> 0,
\end{equation}
in the sense of the QCF
$$
    \bE \re^{i\int_0^T f(t)^\rT X(t)\rd t}
    =
    \re^{-\frac{1}{2} \int_{[0,T]^2} f(s)^\rT P(s-t) f(t)\rd s\rd  t},
    \qquad
    f \in L^2([0,T], \mR^n),
    \
    T >0.
$$
The imaginary part of (\ref{EXX}) is given by (\ref{Lambda}) regardless of the quantum state, while its real part in the invariant Gaussian state is
\begin{equation}
\label{P}
    P(\tau)
     :=
    \left\{
    {\small\begin{matrix}
    \re^{\tau A} \Gamma & {\rm if}\  \tau\> 0\\
    \Gamma \re^{-\tau A^{\rT}} & {\rm if}\  \tau< 0\\
    \end{matrix}}
    \right.
    =
    P(-\tau)^\rT,
    \qquad
    \tau \in \mR,
\end{equation}
where
\begin{equation}
\label{Gamma}
    \Gamma:= \cL_A(BB^\rT)
\end{equation}
is the controllability Gramian of the pair $(A,B)$ in (\ref{ABC}) satisfying the algebraic Lyapunov equation (ALE)
\begin{equation}
\label{PALE}
    A\Gamma + \Gamma A^\rT + BB^\rT = 0.
\end{equation}
In (\ref{Gamma}), use is made of a linear operator $\cL_A$ on $\mC^{n\x n}$, associated with the Hurwitz matrix $A$ as
\begin{equation}
\label{cL}
  \cL_A(V) : = \int_{\mR_+} \re^{tA} V \re^{tA^\rT}\rd t,
  \qquad
  V \in \mC^{n\x n}.
\end{equation}
In comparison with the classical case, the real covariance kernel $P$ in (\ref{P}) satisfies a stronger property of positive semi-definiteness of the quantum covariance kernel $P+i\Lambda$ in (\ref{EXX}). In particular, at any moment of time $t\> 0$,    the one-point quantum covariance matrix of the system variables satisfies
$$
    \bE(X(t)X(t)^\rT)
    =
    \Gamma + i\Theta
    =
    \cL_A(B\Omega B^\rT )
    \succcurlyeq
    0
$$
as the solution of the ALE
\begin{equation*}
\label{PALE1}
    A(\Gamma + i\Theta) + (\Gamma + i\Theta) A^\rT + B\Omega B^\rT = 0,
\end{equation*}
which can be obtained by combining (\ref{PALE}) with the PR condition (\ref{PR1}) and using the quantum Ito matrix $\Omega \succcurlyeq 0$ from (\ref{Omega}).

\section{Quadratic-exponential functional growth rate in frequency domain }
\label{sec:QEF}

Assuming that the matrix $A$ in (\ref{ABC}) is Hurwitz, the input fields are in the vacuum state, and the  OQHO (\ref{dX}), (\ref{dY}) is in the invariant Gaussian quantum state, we are concerned with the infinite-horizon growth rate
\begin{equation}
\label{Upsdef}
    \Ups(\theta):= \lim_{T\to +\infty}
    \Big(
    \frac{1}{T}
    \ln
    \Xi_{\theta,T}
    \Big)
\end{equation}
(whose existence was established in \cite{VPJ_2018a})
for the quadratic-exponential functional  (QEF)
\begin{equation}
\label{Xi}
    \Xi_{\theta,T}
    :=
    \bE
    \re^{\frac{\theta}{2}  \int_0^T X(t)^\rT X(t)\rd t},
\end{equation}
where $\theta\> 0$ is a risk sensitivity parameter. In particular, at $\theta=0$, the QEF reduces to $\Xi_{0,T}=1$, and hence,
\begin{equation}
\label{Upsat0}
    \Ups(0)=0.
\end{equation}
For positive values of $\theta$, the QEF (\ref{Xi}) imposes an exponential penalty on the system variables, so that the minimisation  of $\Xi_{\theta, T}$ (at finite time horizons $T$) or its growth rate $\Ups(\theta)$ secures useful robustness  properties for the quantum system with respect to statistical  uncertainties \cite{VPJ_2018b}  and makes its behaviour more conservative in terms of upper bounds on the tail distributions of system variables \cite{VPJ_2018a}.

Instead of $X^\rT X = \sum_{k=1}^n X_k^2$, the integrand in (\ref{Xi}) can be organised as a more complicated quadratic form $X^\rT V X$, specified by a real positive definite symmetric weighting matrix $V$ of order $n$,  which quantifies the relative importance of the system variables. However, this can be achieved  by applying (\ref{Xi}) to the OQHO resulting from the transformation (\ref{SX})--(\ref{SABC}) with $\sigma := \sqrt{V}$. Furthermore, although the discussion can also be extended to singular weighting matrices $V\succcurlyeq 0$ as limit cases, we will not consider them for simplicity.

The limit value (\ref{Upsdef})  lends itself to computation in frequency domain. More precisely, under the additional constraint
\begin{equation}
\label{BJBdet}
    \det \mho \ne 0,
\end{equation}
on the matrix $\mho$ in (\ref{BJB}),
which implies that $B$ is of full row rank, or equivalently,
\begin{equation}
\label{BBpos}
  BB^\rT \succ 0
\end{equation}
(and hence, $n\< m$ with necessity),
it was found \cite[Theorem~1]{VPJ_2020_IFAC}  (see also \cite{VPJ_2021}) that the QEF growth rate admits the frequency-domain representation
\begin{equation}
\label{Ups}
    \Ups(\theta)
     =
    -
    \frac{1}{4\pi}
    \int_{\mR}
    \ln\det
    D_\theta(\lambda)
    \rd \lambda.
\end{equation}
Here,
\begin{equation}
\label{D}
    D_\theta(\lambda)
    :=
    c_\theta(\lambda)
-
        \Phi(\lambda)
        \Psi(\lambda)^{-1}
        s_\theta(\lambda)
\end{equation}
is a $\mC^{n\x n}$-valued function on $\mR$ involving 
\begin{equation}
\label{cs}
    c_\theta(\lambda)
    :=
      \cos(
        \theta \Psi(\lambda)
    ),
    \qquad
    s_\theta(\lambda)
    :=
      \sin(
        \theta \Psi(\lambda)
    )
\end{equation}
along with the Fourier transforms  of the invariant two-point real covariance and commutator kernels $P$, $\Lambda$ of the system variables in (\ref{EXX}), (\ref{P}), (\ref{Lambda}):
\begin{align}
\nonumber
    \Phi(\lambda)
      & :=
    \int_\mR \re^{-i\lambda t }
    P(t)
    \rd t\\
\label{Phi0}
    & =
    F(i\lambda) F(i\lambda)^*
    =
        E(i\lambda) BB^\rT E(i\lambda)^*, \\
\nonumber
    \Psi(\lambda)
    & :=
    \int_\mR \re^{-i\lambda t }
    \Lambda(t)
    \rd t\\
\label{Psi0}
    & =
    F(i\lambda) J F(i\lambda)^*
    =
    E(i\lambda) \mho E(i\lambda)^*,
    \qquad
    \lambda \in \mR,
\end{align}
where $(\cdot)^*:= {{\overline{(\cdot)}}}^\rT$ is the complex conjugate transpose, and $\mho$ is the matrix from (\ref{BJB}).
The functions $\Phi$, $\Psi$ are related to the $\mC^{n\x m}$-valued real-rational transfer function 
\begin{equation}
\label{F0}
    F(s)
    :=
    E(s)B,
    \qquad
    s \in \mC,
\end{equation}
 from the incremented input quantum Wiener process $W$ of the OQHO (\ref{dX}), (\ref{dY}) to the process $X$. 
Here, use is made of an auxiliary $\mC^{n\x n}$-valued
transfer function
\begin{equation}
\label{Es}
    E(s)
    :=
    (sI_n - A)^{-1},
\end{equation}
which is well defined on the imaginary axis $i\mR$ since $A$ is Hurwitz. Its  state-space realisation
\begin{equation}
\label{Ereal}
    E
    =
      \left[
    \begin{array}{c|c}
    A & I_n\\
      \hline
      I_n &  0
    \end{array}
    \right]
\end{equation}
is identified with a classical  strictly proper linear continuous time invariant (LCTI) system. The conjugate system has the state-space realisation
\begin{equation}
\label{Esimreal}
    E^\sim
    =
      \left[
    \begin{array}{c|c}
    -A^\rT & I_n\\
      \hline
      -I_n &  0
    \end{array}
    \right]
\end{equation}
and the transfer function
\begin{equation}
\label{Esim}
    E^\sim(s)
    =
    -(sI_n + A^\rT)^{-1}.
\end{equation}
In view of (\ref{Omega}), (\ref{Phi0}), (\ref{Psi0}), the Fourier transform of the two-point quantum covariance kernel (\ref{EXX}) of the system variables of the OQHO in the invariant Gaussian quantum state takes the form
\begin{align}
\nonumber
  \Phi(\lambda)+i\Psi(\lambda)
  & =
    \int_\mR
    \re^{-i\lambda t }
    (P(t)+i\Lambda(t))
    \rd t\\
\nonumber
  & =
  F(i\lambda)F(i\lambda)^* + iF(i\lambda) J F(i\lambda)^* \\
\label{QSD}
  & =
  F(i\lambda)\Omega F(i\lambda)^*
  =
  E(i\lambda)B \Omega B^\rT E(i\lambda)^*,
  \qquad
  \lambda \in \mR,
\end{align}
with values in the set $\mH_n^+$ of complex positive semi-definite Hermitian matrices of order $n$,  and plays the role of a quantum spectral density for the process $X$.

The representation (\ref{Ups}) for the QEF growth rate is valid for sufficiently small values of $\theta$ in the sense that
\begin{equation}
\label{spec1}
    \theta
    \sup_{\lambda \in \mR}
    \lambda_{\max}
    (
        \Phi(\lambda)
        \tanc
        (\theta \Psi(\lambda))
    )
    < 1
\end{equation}
(where $\lambda_{\max}(\cdot)$ is the largest eigenvalue of a matrix with a real spectrum)
with the threshold value
\begin{equation}
\label{theta*}
    \theta_*
    :=
    \sup\{\theta > 0:\ (\ref{spec1})\ {\rm is\ satisfied}\}
\end{equation}
(at which (\ref{spec1}) becomes an equality) being different from its classical counterpart
\begin{equation}
\label{theta0}
    \theta_0
    :=
    1\big/\sup_{\lambda \in \mR}
    \lambda_{\max}
    (
        \Phi(\lambda)
    )
    =
    \frac{1}{\|F\|_\infty^2}
\end{equation}
using the $\cH_\infty$-norm $\|F\|_\infty$ of (\ref{F0}). 
As can be seen from \cite[Proof of Theorem~1]{VPJ_2020_IFAC} (see also \cite{VPJ_2021}), the significance of (\ref{spec1}) for well-posedness of the integrand  in (\ref{Ups}) is clarified by
\begin{align}
\label{DD1}
    D_\theta
    & =
    \sqrt{c_\theta }
    \wt{D}_\theta
    \sqrt{c_\theta },\\
\label{DD2}
    \wt{D}_\theta
    & :=
    c_\theta ^{-1/2}
    \sqrt{\Phi }\,
    \wh{D}_\theta
    \Phi ^{-1/2}
    \sqrt{c_\theta },\\
\label{DD3}
    \wh{D}_\theta
    & :=
    I_n - \theta\sqrt{\Phi }\,  \tanc(\theta \Psi )\sqrt{\Phi }
\end{align}
(the dependence on $\lambda$ is omitted for brevity),
where the matrix $\wt{D}_\theta(\lambda)$ is similar to the Hermitian matrix $\wh{D}_\theta(\lambda)$ which is positive definite for any $\lambda \in \mR$  under the condition (\ref{spec1}). Here, use is made of positive definiteness of the matrices $c_\theta(\lambda)$ from (\ref{cs}) and $\Phi(\lambda)$ from (\ref{Phi0}) 
ensured by (\ref{BBpos}) 
and $A$ being Hurwitz (recall that the latter makes the function $E$ in (\ref{Es}) well defined on the imaginary axis).   
Therefore, in view of the invariance of the determinant of a matrix under similarity transformations,  (\ref{DD1})--(\ref{DD3}) lead to
$$
    \ln\det D_\theta(\lambda) = \ln\det c_\theta(\lambda) + \ln\det \wh{D}_\theta(\lambda),
$$
which is well-defined for any $\lambda \in \mR$ under the condition (\ref{spec1}).
Furthermore, as established in  \cite[Theorem~2]{VPJ_2020_IFAC} (see also \cite{VPJ_2021}), the function (\ref{D}) satisfies the second-order linear ODE
\begin{equation}
\label{D''}
  D_\theta(\lambda)''
  + D_\theta(\lambda) \Psi(\lambda)^2 = 0,
  \qquad
  \lambda \in \mR,
\end{equation}
with the initial conditions
$$
    D_0 = I_n,
    \qquad
    D_0' = -\Phi,
$$
where the derivatives  $(\cdot)':= \d_\theta(\cdot)$ and $(\cdot)'':= \d_\theta^2(\cdot)$ are with respect to the risk sensitivity parameter $\theta$. Also, it was obtained there (for a homotopy \cite{MB_1985} algorithm of computing $\Ups$) that
$$
  \Ups'(\theta)
  =
    \frac{1}{4\pi}
    \int_{\mR}
    \Tr U_\theta(\lambda)
    \rd \lambda,
$$
with the initial condition (\ref{Upsat0}), where
\begin{align}
\nonumber
  U_\theta(\lambda)
  := &
   -D_\theta(\lambda)^{-1}D_\theta(\lambda)'\\
\nonumber
  = &
    (    c_\theta(\lambda)
        -
        \Phi(\lambda)
        \Psi(\lambda)^{-1}
        s_\theta(\lambda)
        )^{-1}\\
  \label{U}
        & \x (\Phi(\lambda) c_\theta(\lambda)
    +\Psi(\lambda)s_\theta(\lambda)
    )
\end{align}
is the negative ``logarithmic derivative'' of  $D_\theta(\lambda)$ with respect to $\theta$ satisfying a Riccati ODE
\begin{equation}
\label{U'}
  U_\theta(\lambda)' = \Psi(\lambda)^2 + U_\theta(\lambda)^2,
  \qquad
  \lambda\in \mR,
\end{equation}
with the initial condition
$$
    U_0 = \Phi.
$$
The quadratic nonlinearity on the right-hand side of (\ref{U'}) has a bearing on the trigonometric identities for the $\cos$ and $\sin$ functions which are present in (\ref{D}), (\ref{U}). These functions, evaluated in (\ref{cs}) at the matrix-valued  function $\theta \Psi$,  are related by the Euler identity
\begin{equation}
\label{Euler}
    c_\theta(\lambda) \pm i s_\theta(\lambda)  = \re^{\pm i\theta \Psi(\lambda)} \succ 0
\end{equation}
to two one-parameter matrix groups $\{\re^{\pm i\theta \Psi}: \theta \in \mR\}$ which underlie   the general solution $C_+ \re^{i\theta \Psi(\lambda)} + C_- \re^{-i\theta \Psi(\lambda)}$    of the ODE (\ref{D''}) with arbitrary constant matrices $C_{\pm} \in \mC^{n\x n}$. 
The Hermitian property and positive definiteness of the right-hand side of     (\ref{Euler}) follows from the matrix $\Psi(\lambda)$ in (\ref{Psi0}) being skew Hermitian for any $\lambda \in \mR$.

\section{A spectral density representation of the QEF growth rate}
\label{sec:Delta}

Under the condition (\ref{BJBdet}), the matrix  (\ref{D}) can be expressed in terms of the matrix exponentials $\re^{\pm i\theta \Psi}$ from (\ref{Euler}) as
\begin{align}
\nonumber
    D_\theta
    & :=
        \frac{1}{2}
        (\re^{i\theta \Psi}+ \re^{-i\theta \Psi})
        -
        \Phi
        \Psi^{-1}
        \frac{1}{2i}
        (\re^{i\theta \Psi}- \re^{-i\theta \Psi})\\
\nonumber
    & = \frac{1}{2}
    (I_n + i\Phi\Psi^{-1}) \re^{i\theta \Psi}+
    \frac{1}{2}
    (I_n - i\Phi\Psi^{-1}) \re^{-i\theta \Psi}\\
\nonumber
    & =
    \Big(
        I_n + \frac{1}{2} (I_n +i\Phi \Psi^{-1})(\re^{2i\theta \Psi}-I_n)
    \Big)
    \re^{-i\theta \Psi}\\
\nonumber
    & =
    \Big(
        I_n - \frac{1}{2i} (\Phi - i\Psi)\Psi^{-1}(\re^{2i\theta \Psi}-I_n)
    \Big)
    \re^{-i\theta \Psi}\\
\label{Dexp}
    & =
    (
        I_n -\theta (\Phi - i\Psi)\phi(2i\theta \Psi)
    )
    \re^{-i\theta \Psi},
\end{align}
where
\begin{equation}
\label{phi}
  \phi(u)
  :=
  \left\{
  \begin{matrix}
    1 & {\rm if} & u = 0    \\
    {\frac{\re^{u}-1}{u}} & {\rm if} & u \ne 0
  \end{matrix}
  \right.
  =
  \sum_{k=0}^{+\infty}
  \phi_k u^k,
  \qquad
  u  \in \mC,
\end{equation}
is an entire function with positive values on the real line ($\phi(\mR)\subset (0, +\infty)$) and the coefficients
\begin{equation}
\label{phik0}
    \phi_k:= \frac{1}{(k+1)!},
    \qquad
    k = 0, 1, 2, \ldots.
\end{equation}
In (\ref{Dexp}),   the function $\phi$  is evaluated  \cite{H_2008} at the Hermitian matrix $2i\theta \Psi$ and maps it to a positive definite Hermitian matrix:
\begin{equation}
\label{phipos}
    \phi(2i\theta \Psi(\lambda)) \succ 0 ,
    \qquad
    \lambda \in \mR.
\end{equation}
Now, similarly to (\ref{QSD}),
\begin{equation}
\label{PhiPsiEE}
    \Phi-i\Psi = F\Omega^\rT F^* = EB\Omega^\rT B^\rT E^*,
\end{equation}
where the transfer matrices $F(s)$,  $E(s)$ from (\ref{F0}),  (\ref{Es}) are evaluated at $s:= i\lambda$, with $\lambda \in \mR$, and use is made of
\begin{equation}
\label{OmegaT}
    \Omega^\rT = \overline{\Omega} = I_m - iJ \succcurlyeq 0
\end{equation}
due to the structure of the quantum Ito matrix $\Omega$ in (\ref{Omega}).  Hence, the matrix $B\Omega^\rT B^\rT$ has a square root
\begin{equation}
\label{BOBroot}
  S:= \sqrt{B\Omega^\rT B^\rT} \in \mH_n^+
\end{equation}
and can be replaced with $S^2$.
Since $\det E(i\lambda)\ne 0$ at any frequency $\lambda$, then, in view of (\ref{PhiPsiEE}), the matrix 
\begin{equation}
\label{PhiPsiSig}
    (\Phi - i\Psi)\phi(2i\theta \Psi)
     =
    E S^2 E^* \phi(2i\theta \Psi)
     =
    ES^2
    \Sigma_\theta E^{-1}
\end{equation}
is related by a similarity transformation (and hence, is isospectral) to $S^2\Sigma_\theta$,
where
\begin{equation}
\label{Sigma}
    \Sigma_\theta
    :=
    E^*
     \phi(2i\theta \Psi) E
     =
    E^*
    \sum_{k=0}^{+\infty}
    \phi_k
    (2i\theta E \mho E^* )^k
    E
    \succ 0
\end{equation}
is a complex positive definite Hermitian matrix for any $\lambda \in \mR$. For what follows, we associate with (\ref{Sigma}) a
function $\Delta_\theta: \mR\to \mH_n^+$ by
\begin{equation}
\label{Delta}
    \Delta_\theta(\lambda)
     :=
    S \Sigma_\theta(\lambda)S,
    \qquad
    \lambda \in \mR,
\end{equation}
which has all the properties of the spectral density of a $\mC^n$-valued random process.
Indeed, the Hermitian property and positive semi-definiteness of the matrix $\Delta_\theta(\lambda)$ at any frequency $\lambda$ are inherited from $\Sigma_\theta(\lambda)$ in (\ref{Sigma}) and the Hermitian property of $S$ in (\ref{BOBroot}). Also, due to the transfer function $E$ (and hence, $\Psi$ in (\ref{Psi0})) being strictly proper, and $\phi(0)=1$,  whereby $\|\Delta_\theta(\lambda)\| = o(1/\lambda^2)$ as $\lambda \to \infty$, the function $\Delta_\theta$ is absolutely integrable:
$$
  \int_\mR
  \|\Delta_\theta(\lambda)\|
  \rd \lambda
  <
  +\infty
$$
(this property holds regardless of a particular choice of the matrix norm $\|\cdot\|$). The role of $\Delta_\theta$ for computing the QEF growth rate (\ref{Ups}) is clarified by the following lemma.

\begin{lem}
\label{lem:UpsDel}
Under the condition (\ref{BJBdet}), for any $\theta>0$ subject to (\ref{spec1}), the QEF growth rate (\ref{Ups}) can be represented as
\begin{equation}
\label{UpsDel}
    \Ups(\theta)
     =
    -
    \frac{1}{4\pi}
    \int_{\mR}
    \ln\det
    (I_n - \theta \Delta_\theta(\lambda))
    \rd \lambda
\end{equation}
in terms of (\ref{Delta}). \hfill$\square$
\end{lem}
\begin{proof}
From the representation (\ref{Dexp}) of the function $D_\theta$ in (\ref{D}), it follows that
\begin{align}
\nonumber
    \int_{\mR}
    \ln\det
    D_\theta(\lambda)
    \rd \lambda
     =&
    \int_{\mR}
    \ln\det
    (
        I_n -\theta (\Phi(\lambda) - i\Psi(\lambda))\phi(2i\theta \Psi(\lambda))
    )
    \rd \lambda\\
\nonumber
    & +    \int_\mR
    \ln\det \re^{-i\theta \Psi(\lambda)}\rd \lambda\\
\nonumber
    = &
    \int_{\mR}
    \ln
    \det
    (
        I_n -
        \theta
        E(i\lambda)S^2\Sigma_\theta(\lambda) E(i\lambda)^{-1}
    )
    \rd \lambda\\
\label{UpsDel1}
    = &
    \int_{\mR}
    \ln
    \det
    (
        I_n -
        \theta
        S^2\Sigma_\theta(\lambda)
    )
    \rd \lambda.
\end{align}
Here, use is made of the relation (\ref{PhiPsiSig}) (along with the isospectrality mentioned in regard to it) and
$$
    \int_\mR
    \ln\det \re^{-i\theta \Psi(\lambda)}\rd \lambda
    =
    -i\theta
    \int_\mR
    \Tr \Psi(\lambda)
    \rd \lambda
    =
    -2\pi i\theta\Tr \Lambda(0) = 0.
$$
The last equality is obtained by taking the trace of the inverse Fourier transform applied to (\ref{Psi0}) as
$$
    \Tr \Lambda(\tau) = \frac{1}{2\pi} \int_\mR \re^{i\lambda \tau} \Tr \Psi(\lambda)\rd \lambda,
    \qquad
    \tau \in \mR,
$$
and using (\ref{Lambda0}) together with the antisymmetry of the CCR matrix $\Theta$ (which makes  it traceless: $\Tr \Theta =0$).  
Since 
the matrix $S^2\Sigma_\theta(\lambda)$ 
is isospectral to
$\Delta_\theta(\lambda)$ in (\ref{Delta}), 
 then
\begin{equation}
\label{detdet}
    \det
    (
        I_n -\theta
        S^2\Sigma_\theta(\lambda)
    )
    =
    \det (I_n - \theta \Delta_\theta(\lambda)).
\end{equation}
In view of (\ref{DD1})--(\ref{DD3}), the condition (\ref{spec1}) is equivalent to $\theta \Delta_\theta(\lambda) \prec I_n$ for all $\lambda \in \mR$, thus making the logarithm  well-defined in application to (\ref{detdet}). Substitution of (\ref{detdet}) into (\ref{UpsDel1}) leads to the representation (\ref{UpsDel}) for the QEF growth rate (\ref{Ups}).
\end{proof}

The representation (\ref{UpsDel}) will be used in Sections~\ref{sec:inf}, \ref{sec:class} in order to relate the quantum QEF growth rate to a similar functional for a classical Gaussian random process.

\section{Infinite cascade spectral factorization}
\label{sec:inf}

The state-space realisations (\ref{Ereal}), (\ref{Esimreal}) of the systems $E$, $E^\sim$ allow $\Psi$ in (\ref{Psi0}) (as a function of $i\lambda$ rather than $\lambda$, with a slight abuse of notation) to be identified with the transfer function of the LCTI system
\begin{equation}
\label{Psireal}
    \Psi
    =
    E\mho E^\sim
    =
    \left[
    \begin{array}{cc|c}
      -A^\rT & 0 & I_n\\
      -\mho & A & 0\\
      \hline
      0 & I_n & 0
    \end{array}
    \right],
\end{equation}
where the matrix $\mho$ is given by (\ref{BJB}). The dynamics matrix of this state-space realisation can be  block diagonalised by the similarity transformation
\begin{align}
\nonumber
  \begin{bmatrix}
    I_n & 0 \\
    \Theta & I_n
  \end{bmatrix}
  \begin{bmatrix}
    -A^\rT & 0\\
    -\mho & A
  \end{bmatrix}
  \begin{bmatrix}
    I_n & 0 \\
    -\Theta & I_n
  \end{bmatrix}
    & =
  \begin{bmatrix}
    -A^\rT & 0\\
    -\Theta A^\rT-\mho & A
  \end{bmatrix}
  \begin{bmatrix}
    I_n & 0 \\
    -\Theta & I_n
  \end{bmatrix}\\
\label{block2}
    & =
  \begin{bmatrix}
    -A^\rT & 0\\
    A \Theta & A
  \end{bmatrix}
  \begin{bmatrix}
    I_n & 0 \\
    -\Theta & I_n
  \end{bmatrix}=
  \begin{bmatrix}
    -A^\rT & 0\\
    0 & A
  \end{bmatrix},
\end{align}
which uses  the CCR matrix $\Theta$ from (\ref{XCCR}) along with the PR condition (\ref{PR1}) and the identity
$$
  \begin{bmatrix}
    I_n & 0 \\
    \alpha & I_n
  \end{bmatrix}
  \begin{bmatrix}
    I_n & 0 \\
    \beta & I_n
  \end{bmatrix}
  =
  \begin{bmatrix}
    I_n & 0 \\
    \alpha + \beta & I_n
  \end{bmatrix},
  \qquad
  \alpha, \beta \in \mC^{n\x n}
$$
(a homomorphism between the multiplicative group of block lower triangular matrices with identity diagonal blocks and $\mC^{n\x n}$ as an additive group). Therefore, (\ref{Psireal})  admits an equivalent state-space realisation:
\begin{equation}
\label{Psireal1}
    \Psi
    =
    \left[
    \begin{array}{c|c}
      a & b\\
      \hline
      c & 0
    \end{array}
    \right],
    \quad
        a
    :=
    \begin{bmatrix}
        -A^\rT & 0\\
  0 & A
    \end{bmatrix},
    \quad
    b:=
    \begin{bmatrix}
      I_n\\
      \Theta
    \end{bmatrix},
    \quad
    c:=
    \begin{bmatrix}
      -\Theta & I_n
    \end{bmatrix}.
\end{equation}
With $\Psi$ being interpreted as the transfer function of a finite-dimensional LCTI system with the strictly proper   state-space realisation (\ref{Psireal}) (or (\ref{Psireal1})),
the matrix $\phi(2i\theta \Psi(\lambda))$ in (\ref{phipos}) corresponds to the transfer  function for an infinite cascade of such systems shown in Fig.~\ref{fig:phi}. 
\begin{figure}[htbp]
\centering
\unitlength=1.15mm
\linethickness{0.5pt}
\begin{picture}(55.00,39.00)
    \put(-15,0){\dashbox(85,37)[cc]{}}
    \put(-14,36){\makebox(0,0)[lt]{{\small$\phi(2i\theta \Psi)$}}}
    \put(11,26){\framebox(8,8)[cc]{\small$\Psi$}}
    \put(31,26){\framebox(8,8)[cc]{\small$\Psi$}}
    \put(51,26){\framebox(8,8)[cc]{\small$\Psi$}}

    \put(11,30){\vector(-1,0){12}}
    \put(51,30){\vector(-1,0){12}}
    \put(31,30){\vector(-1,0){12}}
    \put(75,30){\vector(-1,0){16}}
    \put(65,30){\line(0,-1){25}}
    \put(45,30){\vector(0,-1){8}}
    \put(25,30){\vector(0,-1){8}}
    \put(5,30){\vector(0,-1){8}}
    \put(41,14){\framebox(8,8)[cc]{\small$i\theta$}}
    \put(21,14){\framebox(8,8)[cc]{\small$-\frac{2}{3}\theta^2$}}
    \put(1,14){\framebox(8,8)[cc]{\small$-\frac{i}{3}\theta^3$}}
    \put(65,14){\line(0,-1){9}}
    \put(45,14){\vector(0,-1){6.3}}
    \put(25,14){\vector(0,-1){6.3}}
    \put(5,14){\vector(0,-1){6.3}}
    \put(45,5){\circle{5}}
    \put(25,5){\circle{5}}
    \put(5,5){\circle{5}}
    \put(65,5){\vector(-1,0){17.3}}
    \put(42.3,5){\vector(-1,0){14.5}}
    \put(22.3,5){\vector(-1,0){14.5}}
    \put(-10,5){\vector(-1,0){10}}
    \put(45,5){\makebox(0,0)[cc]{{$+$}}}
    \put(25,5){\makebox(0,0)[cc]{{$+$}}}
    \put(5,5){\makebox(0,0)[cc]{{$+$}}}
    \put(-2,30){\makebox(0,0)[rc]{{$\cdots$}}}
    \put(-2,5){\makebox(0,0)[rc]{{$\cdots$}}}
    \put(77,30){\makebox(0,0)[lc]{{\small$\omega$}}}
    \put(-25,5){\makebox(0,0)[lc]{{\small$\zeta$}}}
\end{picture}
\caption{An infinite cascade of identical LCTI systems with the state-space realization of $\Psi$ in  (\ref{Psireal}), (\ref{Psireal1}) and an external input $\omega$. The sum of their outputs, weighted by $(2i\theta)^k/(k+1)!$  in accordance with the coefficients (\ref{phik0}) of the entire function $\phi$ from (\ref{phi}), forms the output $\zeta$ of the LCTI system $\phi(2i\theta \Psi)$ in (\ref{phipos}).
}
\label{fig:phi}
\end{figure}
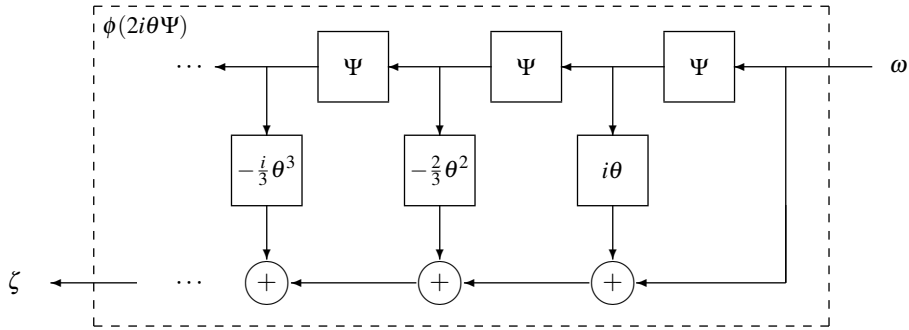 
The resulting LCTI system  has an infinite dimensional state and the state-space realisation 
\begin{equation}
\label{phitrans}
    \phi(z \Psi)
    =
    \left[\begin{array}{cccc|c}
  a   & 0 & 0 &\ldots & b\\
  b c & a & 0 & \ldots & 0\\
  0   & bc & a & \ldots & 0\\
  \ldots    &\ldots   & \ldots  & \ldots & \ldots  \\
  \hline
  \phi_1 z c & \phi_2 z^2 c & \phi_3 z^3c & \cdots & \phi_0 I_n
\end{array}\right],
\qquad
z:= 2i\theta,
\end{equation}
where $a$, $b$, $c$ 
are the state-space matrices of $\Psi$ in (\ref{Psireal1}). In accordance with the cascade structure of these systems,  (\ref{phitrans}) has a  block  two-diagonal   lower triangular dynamics matrix. It will be shown below that the spectral densities $\Sigma_\theta$, $\Delta_\theta$ 
in (\ref{Sigma}), (\ref{Delta}), 
which involve (\ref{phitrans}),  admit an inner-outer  factorization (see, for example, \cite{W_1972}) with a similar infinite cascade structure. To this end, we will need the following two lemmas and a theorem, which develop a ``system transposition'' technique.

\begin{lem}
\label{lem:EOE}
Let $U \in \mR^{n\x n}$ be an  arbitrary matrix  such that the solution $V$ of the ALE
\begin{equation}
\label{ZALE}
    AV + VA^\rT + U = 0,
\end{equation}
associated with the Hurwitz matrix $A$ from (\ref{ABC}), is nonsingular. 
Then the transfer function $E$ in (\ref{Es}) and its system conjugate $E^\sim$ in (\ref{Esim}) satisfy
\begin{equation}
\label{EOE}
    E(s) U E^\sim(s) = V E^\sim(s)  V^{-1}U V^{-1} E(s)V
\end{equation}
for any $s\in \mC$ not belonging to the spectra of $\pm A$.
\hfill$\square$
\end{lem}
\begin{proof}
By using the ALE (\ref{ZALE}) and omitting the argument $s \in \mC$ (with values beyond the spectra of $\pm A$)  of the transfer functions $E$, $E^\sim$ for brevity, it follows that
\begin{align}
\nonumber
    E U  E^\sim
    & =
    -E(A V  +  V  A^\rT)E^\sim      \\
\nonumber
    & =
    -E((A-sI_n) V  +  V  (A^\rT+sI_n))E^\sim    \\
\nonumber
    & =  V  E^\sim +  E V     \\
\nonumber
    & =  V  E^\sim( V ^{-1}E^{-1} +  (E^\sim)^{-1} V ^{-1})E V     \\
\nonumber
    & =  V  E^\sim( V ^{-1}(sI_n-A) -  (sI_n+A^\rT) V ^{-1})E V     \\
\nonumber
    & = - V  E^\sim( V ^{-1}A +  A^\rT  V ^{-1})E V     \\
\label{EUE}
    & =  V  E^\sim   V ^{-1} U   V ^{-1} E V  ,
\end{align}
which establishes (\ref{EOE}).
The last equality in (\ref{EUE}) employs the fact that, under the condition $\det  V  \ne 0$,   the ALE (\ref{ZALE}) is equivalent to an algebraic  Riccati equation (ARE) 
$$
     V ^{-1}A + A^\rT  V ^{-1}  +  V ^{-1}U V ^{-1} = 0
$$
obtained by the left and right multiplication of (\ref{ZALE}) by $ V ^{-1}$.
\end{proof}

The relation (\ref{EOE}) is equivalent to
$$
    V^{-1}E U E^\sim V^{-1} - E^\sim  V^{-1}U V^{-1} E = 0,
$$
which can be interpreted as commutativity of $V^{-1}E $ and $E^\sim V^{-1}$ ``through'' the matrix $U$.  Also, an alternative way to obtain (\ref{EOE}) is to use the property that the dynamics matrix of
$$
    E U   E^\sim
    =
    \left[
    \begin{array}{cc|c}
      -A^\rT & 0 & I_n\\
      - U   & A & 0\\
      \hline
      0 & I_n & 0
    \end{array}
    \right]
$$
is block diagonalisable as
\begin{align*}
  \begin{bmatrix}
    I_n & 0 \\
     V   & I_n
  \end{bmatrix}
  \begin{bmatrix}
    -A^\rT & 0\\
    - U   & A
  \end{bmatrix}
  \begin{bmatrix}
    I_n & 0 \\
    - V   & I_n
  \end{bmatrix}
    & =
  \begin{bmatrix}
    -A^\rT & 0\\
    - V   A^\rT- U   & A
  \end{bmatrix}
  \begin{bmatrix}
    I_n & 0 \\
    - V   & I_n
  \end{bmatrix}\\
    & =
  \begin{bmatrix}
    -A^\rT & 0\\
    A  V   & A
  \end{bmatrix}
  \begin{bmatrix}
    I_n & 0 \\
    - V   & I_n
  \end{bmatrix}=
  \begin{bmatrix}
    -A^\rT & 0\\
    0 & A
  \end{bmatrix}
\end{align*}
similarly to (\ref{block2}),
and hence,
$$
        E U   E^\sim
        =
        \left[
    \begin{array}{cc|c}
      -A^\rT & 0 & I_n\\
      0 & A &  V  \\
      \hline
      - V   & I_n & 0
    \end{array}
    \right]
    =
        \left[
    \begin{array}{cc|c}
      A & 0 &  V  \\
      0 & -A^\rT & I_n\\
      \hline
      I_n & - V   & 0
    \end{array}
    \right],
$$
where the second equality is obtained by swapping the decoupled subsystems $(A,V,I_n)$ and $(-A^\rT,I_n,-V)$ (without affecting their  sum as linear operators with a common input), which leads to the right-hand side of (\ref{EOE}).

In application of Lemma~\ref{lem:EOE} to the matrix $U:= \mho$ from (\ref{BJB}), the ALE (\ref{ZALE}) coincides with the PR condition (\ref{PR1}) and yields
the CCR matrix $V:= \Theta$ in (\ref{XCCR}) which is assumed to be nonsingular. In this case, the relation  (\ref{EOE}) allows the factors in (\ref{Psireal}) to be rearranged as
\begin{equation}
\label{EmhoE}
  E \mho E^\sim = \Theta E^\sim  \Theta^{-1}\mho \Theta^{-1} E\Theta .
\end{equation}
In the representation (\ref{EOE}) and its particular case (\ref{EmhoE}), the factors $E$ are moved to the right, while their duals $E^\sim$ are moved to the left. This resembles the Wick ordering \cite{W_1950} (see also \cite[pp. 209--210]{J_1997}) for mixed products of quantum mechanical annihilation and creation operators. Theorem~\ref{th:EOEk} below extends (\ref{EmhoE}) to arbitrary positive integer powers of $E \mho E^\sim$. Its formulation employs three sequences of matrices $\alpha_j, \beta_j, \gamma_j  \in \mR^{n\x n}$ computed recursively as
\begin{equation}
\label{alfbetgamnext}
    \alpha_{j+1} = \gamma_j \beta_j,
    \quad
    \beta_{j+1} = \gamma_j^{-1} \alpha_j \gamma_{j-1}\gamma_j^{-1},
    \quad
    \gamma_j= \cL_A(\alpha_j \gamma_{j-1}),
    \qquad
    j\> 1
\end{equation}
(where the operator $\cL_A$ is given by (\ref{cL}))
with the initial conditions
\begin{equation}
\label{alfbetgam0}
        \alpha_1
    =
    \gamma_0
    =
    \cL_A(\mho)
    =
    \Theta,
    \qquad
    \beta_1 = \Theta^{-1}\mho \Theta^{-1}.
\end{equation}
It is convenient to extend the first equality in (\ref{alfbetgamnext}) to $j=0$ as $\alpha_1 = \gamma_0 \beta_0$ by letting
\begin{equation}
\label{bet0}
    \beta_0:= \gamma_0^{-1} \alpha_1 = I_n,
\end{equation}
 in accordance with the first equality from (\ref{alfbetgam0}). Also, it is assumed that the matrices $\gamma_j$ in (\ref{alfbetgamnext}) are nonsingular:
 \begin{equation}
 \label{gamdet}
   \det \gamma_j \ne 0,
   \qquad
   j\> 1.
 \end{equation}
The following lemma provides relevant properties of these matrices which will be used in the  proof of Theorem~\ref{th:EOEk}.

\begin{lem}
\label{lem:alfbetgam}
The matrices $\beta_j$, $\gamma_j$, defined by (\ref{alfbetgamnext})--(\ref{bet0}) subject to (\ref{gamdet}),   have the opposite symmetric properties:
\begin{equation}
\label{betgam+-}
    \beta_j^\rT = (-1)^j \beta_j,
    \qquad
    \gamma_j^\rT = -(-1)^j \gamma_j,
    \qquad
    j \> 0.
\end{equation}
Furthermore, the matrices $\beta_j$ with even $j$ have an alternating definiteness in the sense that
\begin{equation}
\label{betpos}
    (-1)^r\beta_{2r} \succ 0,
    \qquad
    r \> 0.
\end{equation}
\hfill$\square$
\end{lem}
\begin{proof}
From the recurrence relations (\ref{alfbetgamnext}), it follows  that
\begin{equation}
\label{betgamnext}
    \beta_{j+1} = \gamma_j^{-1} \gamma_{j-1} \beta_{j-1} \gamma_{j-1}\gamma_j^{-1},
    \qquad
    \gamma_j= \cL_A(\gamma_{j-1} \beta_{j-1}  \gamma_{j-1}),
    \qquad
    j\> 1,
\end{equation}
which, in view of the initial conditions (\ref{alfbetgam0}), (\ref{bet0}),  implies by induction that the matrices $\beta_{2r}$, $\gamma_{2r+1}$ are symmetric, while $\beta_{2r+1}$, $\gamma_{2r}$ are antisymmetric for all $r = 0,1,2,\ldots$, thus establishing (\ref{betgam+-}).  Here, use is made of the commutativity between the operator $\cL_A$ and the matrix transpose ($\cL_A(\beta^\rT) = \cL_A(\beta)^\rT$ for any $\beta \in \mR^{n\x n}$),  whereby the subspaces $\mS_n$, $\mA_n$ of real symmetric and real antisymmetric matrices of order $n$ are invariant under $\cL_A$ in (\ref{cL}):
\begin{equation}
\label{incs}
    \cL_A(\mS_n) \subset \mS_n,
    \qquad
    \cL_A(\mA_n) \subset \mA_n.
\end{equation}
Now, the validity of (\ref{betpos}) at $r=0$ is secured by (\ref{bet0}). From the second relation in (\ref{betgam+-}), it follows that $(\gamma_j^{-1} \gamma_{j-1})^\rT =
-\gamma_{j-1}\gamma_j^{-1}$ for any $j$, whereby
(\ref{betgamnext}) implies that
$$
    (-1)^{r+1}\beta_{2r+2} = \gamma_{2r+1}^{-1} \gamma_{2r} (-1)^r\beta_{2r} (\gamma_{2r+1}^{-1} \gamma_{2r} )^\rT
    \succ 0,
$$
provided (\ref{betpos}) is already proved for some $r\> 0$.
Hence, by induction, (\ref{betpos}) holds for any $r\> 0$.
\end{proof}

Note that (\ref{betpos}) is closely related to the ``propagation'' of nonsingularity (\ref{gamdet}) over odd values of $j$. More precisely, in view of (\ref{betgam+-}), by applying the second equality from (\ref{betgamnext}) to $j=2r+1$, it follows that a combination of (\ref{betpos}) with  $\det \gamma_{2r} \ne 0$ leads to
\begin{equation}
\label{prop}
    -(-1)^r\gamma_{2r+1}=
    \cL_A(\gamma_{2r} (-1)^r\beta_{2r}  \gamma_{2r}^\rT) \succ 0,
\end{equation}
and hence, $\det \gamma_{2r+1}\ne 0$. Here, use is also  made of the inclusion (due to $A$ being Hurwitz)
\begin{equation}
\label{inc1}
    \cL_A(\mP_n) \subset \mP_n,
\end{equation}
which complements (\ref{incs}),
with $\mP_n$ the set of real positive definite symmetric matrices of order $n$. However, the inequality (\ref{prop}), which  is based on the definiteness argument using (\ref{inc1}), does not extend to even values of $j$, thus explaining the need in the assumption (\ref{gamdet}).

\begin{thm}
\label{th:EOEk}
The system $E$ in (\ref{Ereal}) and  its dual $E^\sim$ in (\ref{Esimreal}), associated with the Hurwitz matrix $A$,   satisfy
\begin{equation}
\label{EOEk}
    (E\mho E^\sim)^k
    =
    (-1)^k
    \rprod_{j=1}^k
    (\alpha_j^\rT E^\sim)
    \beta_k
    \lprod_{j=1}^k
    (E \alpha_j),
    \qquad
    k \> 1,
\end{equation}
where $\mho$ is the matrix from  (\ref{BJB}), and $\rprod(\cdot)$, $\lprod(\cdot)$ are the rightwards and leftwards  ordered products, respectively. Here,
the matrices $\alpha_j, \beta_j \in \mR^{n\x n}$ are given by (\ref{alfbetgamnext}), (\ref{alfbetgam0}) subject to the condition (\ref{gamdet}). \hfill$\square$ 
\end{thm}
\begin{proof}
We will use a nested induction over $k\>1$ (the outer layer of induction) and $j= 1, \ldots, k-1$ (the inner layer, which is inactive at $k=1$).
The validity of (\ref{EOEk}) at $k=1$,  that is,
\begin{equation}
\label{EOE1}
  E\mho E^\sim
  =
  -\alpha_1^\rT E^\sim \beta_1 E \alpha_1
  =
  \gamma_0 E^\sim \beta_1 E \alpha_1
\end{equation}
with the matrices $\alpha_1$, $\beta_1$ given by (\ref{alfbetgam0}), follows from (\ref{EmhoE}) in view of the antisymmetry and nonsingularity of the CCR matrix $\Theta$ from (\ref{XCCR}). In order to demonstrate  the structure of the induction steps, consider the left-hand side of (\ref{EOEk}) for $k=2$:
\begin{align}
\nonumber
  (E \mho E^\sim)^2
  & =
  E \mho E^\sim E\mho E^\sim\\
\nonumber
  & = -\alpha_1^\rT E^\sim \beta_1 E \alpha_1 E\mho E^\sim\\
\nonumber
  & = -\alpha_1^\rT E^\sim \beta_1 E \alpha_1 \gamma_0 E^\sim\beta_1 E\alpha_1\\
\label{EOE2}
  & =
  -\alpha_1^\rT
  E^\sim
  \underbrace{\beta_1 \gamma_1}_{-\alpha_2^\rT}
  E^\sim
  \underbrace{\gamma_1^{-1} \alpha_1 \gamma_0 \gamma_1^{-1}}_{\beta_2}
  E
  \underbrace{\gamma_1\beta_1}_{\alpha_2}
  E
  \alpha_1,
\end{align}
where (\ref{alfbetgamnext}) is used for $j=1$ along with (\ref{alfbetgam0}), (\ref{EOE1}),   the antisymmetry of $\beta_1$ and symmetry of
$$
    \gamma_1 = \cL_A(\alpha_1\gamma_0) = \cL_A(\Theta^2).
$$
The last two equalities in (\ref{EOE2}) involve repeated application of (\ref{EOE}) from Lemma~\ref{lem:EOE}. This allows the rightmost $E^\sim$ factor to be ``pulled'' through the product of the $E$ factors (and constant matrices between them) until the $E^\sim$ factor is to the left of all the $E$ factors. The last equality in (\ref{EOE2}) also uses the symmetric properties (\ref{betgam+-}) and  establishes (\ref{EOEk}) for $k=2$. Now, suppose the representation (\ref{EOEk}) is already proved for some $k \> 2$. Then the next power of the matrix $E\mho E^\sim $ takes the form
\begin{equation}
\label{EOEnext}
    (E\mho E^\sim )^{k+1}
    =
    (-1)^k
    \rprod_{j=1}^k
    (\alpha_j^\rT E^\sim )
    \beta_k
    \lprod_{j=1}^k
    (E \alpha_j)
    E\mho E^\sim ,
\end{equation}
where the rightmost $E^\sim$ factor is the only $E^\sim$ factor which is to the right of the $E$ factors. We will now use the pulling procedure, demonstrated in (\ref{EOE2}),  and prove that
\begin{equation}
\label{inner}
    \beta_k
    \lprod_{j=1}^k
    (E \alpha_j)
    E\mho E^\sim
    =
    \beta_k
    \lprod_{j=r}^k
    (E \alpha_j)
    \gamma_{r-1}
    E^\sim
    \beta_r
    \lprod_{j=1}^r
    (E \alpha_j)
\end{equation}
by induction over $r= 1, \ldots, k$.  The fulfillment of (\ref{inner}) for $r=1$ is verified by applying (\ref{EmhoE}) and using (\ref{alfbetgam0}):
\begin{align*}
    \beta_k
    \lprod_{j=1}^k
    (E \alpha_j)
    E\mho E^\sim
    & =
    \beta_k
    \lprod_{j=1}^k
    (E \alpha_j)
    E\mho E^\sim\\
    & =
    \beta_k
    \lprod_{j=1}^k
    (E \alpha_j)
    \Theta
    E^\sim
    \Theta^{-1} \mho \Theta^{-1}
    E
    \Theta\\
    & =
    \beta_k
    \lprod_{j=1}^k
    (E \alpha_j)
    \gamma_0
    E^\sim
    \beta_1
    E
    \alpha_1.
\end{align*}
Now, suppose (\ref{inner}) is already proved for some $r= 1, \ldots, k-1$. Then its validity for the next value $r+1$ is established by using (\ref{EOE}) of Lemma~\ref{lem:EOE} in combination with (\ref{alfbetgamnext}) as
\begin{align*}
    \beta_k
    \lprod_{j=1}^k
    (E \alpha_j)
    E\mho E^\sim
    & =
    \beta_k
    \lprod_{j=r+1}^k
    (E \alpha_j)
    E \alpha_r
    \gamma_{r-1}
    E^\sim
    \beta_r
    \lprod_{j=1}^r
    (E \alpha_j)\\
    & =
    \beta_k
    \lprod_{j=r+1}^k
    (E \alpha_j)
    \gamma_r
    E^\sim
    \underbrace{\gamma_r^{-1}\alpha_r
    \gamma_{r-1}
    \gamma_r^{-1}}_{\beta_{r+1}}
    E
    \underbrace{\gamma_r
    \beta_r}_{\alpha_{r+1}}
    \lprod_{j=1}^r
    (E \alpha_j)\\
    & =
    \beta_k
    \lprod_{j=r+1}^k
    (E \alpha_j)
    \gamma_r
    E^\sim
    \beta_{r+1}
    \lprod_{j=1}^{r+1}
    (E \alpha_j).
\end{align*}
Therefore, (\ref{inner}) holds for any $r = 1, \ldots, k$, thus completing the inner layer of induction. In particular, at $r=k$, this relation takes the form
\begin{align}
\nonumber
    \beta_k
    \lprod_{j=1}^k
    (E \alpha_j)
    E\mho E^\sim
    & =
    \beta_k
    E
    \alpha_k\gamma_{k-1}
    E^\sim
    \beta_k
    \lprod_{j=1}^k
    (E \alpha_j)\\
\nonumber
    & =
    \underbrace{\beta_k
    \gamma_k}_{-\alpha_{k+1}^\rT}
    E^\sim
    \underbrace{\gamma_k^{-1}
    \alpha_k\gamma_{k-1}
    \gamma_k^{-1}}_{\beta_{k+1}}
    E
    \underbrace{\gamma_k
    \beta_k}_{\alpha_{k+1}}
    \lprod_{j=1}^k
    (E \alpha_j)\\
\label{innerlast}
    & =
    -\alpha_{k+1}^\rT
    E^\sim
    \beta_{k+1}
    \lprod_{j=1}^{k+1}
    (E \alpha_j),
\end{align}
where (\ref{EOE}) of Lemma~\ref{lem:EOE} and (\ref{alfbetgamnext}) are used again along with $\alpha_{k+1}^\rT = \beta_k^\rT \gamma_k^\rT = -\beta_k\gamma_k$ in view of (\ref{betgam+-}). Now, substitution of (\ref{innerlast}) into (\ref{EOEnext}) leads to
\begin{align*}
    (E\mho E^\sim )^{k+1}
    & =
    -(-1)^k
    \rprod_{j=1}^k
    (\alpha_j^\rT E^\sim )
    \alpha_{k+1}^\rT
    E^\sim
    \beta_{k+1}
    \lprod_{j=1}^{k+1}
    (E \alpha_j)\\
    & =
    (-1)^{k+1}
    \rprod_{j=1}^{k+1}
    (\alpha_j^\rT E^\sim )
    \beta_{k+1}
    \lprod_{j=1}^{k+1}
    (E \alpha_j),
\end{align*}
which completes the outer layer of induction, thus proving (\ref{EOEk}) for any $k\>1$.
\end{proof}

For any $\sigma_1, \sigma_2, \sigma_3, \ldots \in \mR\setminus\{0\}$,   the factorisation (\ref{EOEk}) is invariant under the transformation
$$
    \alpha_k \mapsto \frac{1}{\sigma_k} \alpha_k,
    \qquad
    \beta_k \mapsto \beta_k \prod_{j=1}^k \sigma_j^2,
    \qquad
    k \> 1,
$$
which can be used for balancing the state-space realisations discussed below.
The following theorem employs an extension of Theorem~\ref{th:EOEk} from monomials to entire functions of  $E\mho E^\sim$ along with the duality
\begin{equation}
\label{EEEdual}
        \rprod_{j=1}^k
    (\alpha_j^\rT E^\sim)
    =
    \Big(\lprod_{j=1}^k
    (E \alpha_j)\Big)^\sim.
\end{equation}
To this end,
the system $E$ and the sequence of matrices $\alpha_k$  give rise to strictly proper LCTI systems
\begin{equation}
\label{Gk}
    G_k :=
    \lprod_{j=1}^k
    (E \alpha_j)
    E
    =
    \lprod_{j=0}^k
    (E \alpha_j),
    \qquad
    k = 0,1,2,\ldots
\end{equation}
(with the second equality using the convention that $\alpha_0=I_n$), which are assembled into
\begin{equation}
\label{G}
    G
    :=
    \begin{bmatrix}
      G_0\\
      G_1\\
      G_2\\
      \vdots
    \end{bmatrix}
    =
    \begin{bmatrix}
      E\\
      E \alpha_1 E\\
      E \alpha_2 E \alpha_1 E\\
      \vdots
    \end{bmatrix}
    =
    \left[\begin{array}{cccc|c}
  A   & 0 & 0 &\ldots & I_n\\
  \alpha_1  & A & 0 & \ldots  & 0\\
  0   & \alpha_2 & A & \ldots  & 0\\
  \ldots    &\ldots   & \ldots  & \ldots  &\ldots  \\
  \hline
  I_n & 0 & 0 & \ldots  & 0\\
  0 & I_n & 0 & \ldots  & 0\\
  0 & 0 & I_n & \ldots  & 0\\
  \ldots & \ldots & \ldots & \ldots  & \ldots
\end{array}\right],
\end{equation}
provided the condition (\ref{gamdet}) is satisfied.
The system $G$ has real state-space matrices  (its output matrix is the infinite identity matrix $I_\infty$),
an $\mR^n$-valued input and an $\mR^\infty $-valued output which coincides with the internal state. This system is organised as an infinite cascade of LCTI systems shown in Fig.~\ref{fig:EEE}. 
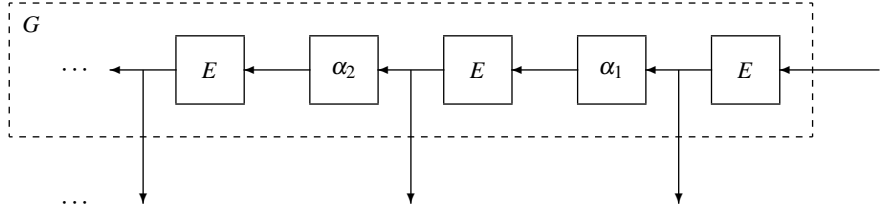
\begin{figure}[htbp]
\centering
\unitlength=0.89mm
\linethickness{0.5pt}
\begin{picture}(140.00,36.00)
    \put(5,10){\dashbox(120,20)[cc]{}}
    \put(7,28){\makebox(0,0)[lt]{{\small$G$}}}
    \put(110,15){\framebox(10,10)[cc]{\small$E$}}
    \put(90,15){\framebox(10,10)[cc]{\small$\alpha_1$}}
    \put(70,15){\framebox(10,10)[cc]{\small$E$}}
    \put(50,15){\framebox(10,10)[cc]{\small$\alpha_2$}}
    \put(30,15){\framebox(10,10)[cc]{\small$E$}}
    \put(135,20){\vector(-1,0){15}}
    \put(110,20){\vector(-1,0){10}}
    \put(90,20){\vector(-1,0){10}}
    \put(70,20){\vector(-1,0){10}}
    \put(50,20){\vector(-1,0){10}}
    \put(30,20){\vector(-1,0){10}}
    \put(15,20){\makebox(0,0)[cc]{{$\cdots$}}}
    \put(15,0){\makebox(0,0)[cc]{{$\cdots$}}}

    \put(105,20){\vector(0,-1){20}}
    \put(65,20){\vector(0,-1){20}}
    \put(25,20){\vector(0,-1){20}}
\end{picture}
\caption{An infinite cascade of copies of the LCTI system $E$ from (\ref{Ereal}) (with the matrices $\alpha_k$ from (\ref{alfbetgamnext}), (\ref{alfbetgam0}) as intermediate factors)  forming the system $G$ with the state-space realization (\ref{G}).
}
\label{fig:EEE}
\end{figure} 
We will also use an infinite block-diagonal complex Hermitian matrix
\begin{equation}
\label{H}
    H_\theta
    :=
    \diag_{k\> 0}
        ((-2i\theta)^k
    \phi_k \beta_k),
\end{equation}
defined in terms of the coefficients $\phi_k$ in  (\ref{phik0}) and the matrices $\beta_k$ from Lemma~\ref{lem:alfbetgam}.

\begin{thm}
\label{th:Delta}
Under the condition (\ref{gamdet}), the spectral density $\Delta_\theta$ in  (\ref{Delta}) can be  factorised in terms of the transfer function of the system $G$ from  (\ref{G}), the matrix $H_\theta$ in  (\ref{H}) and the square root $S$ in (\ref{BOBroot}) as
\begin{equation}
\label{DelGH}
  \Delta_\theta(\lambda)
  =
  SG(i\lambda)^* H_\theta G(i\lambda)S,
  \qquad
  \lambda \in \mR.
\end{equation}
\hfill$\square$
\end{thm}
\begin{proof}
By considering (\ref{EOEk}) on the imaginary axis $i\mR$  and substituting it into (\ref{Sigma}), it follows that
\begin{align}
\nonumber
    \Sigma_\theta(\lambda)
    & =
    E(i\lambda)^*
    \sum_{k=0}^{+\infty}
    (-2i\theta)^k
    \phi_k
    \rprod_{j=1}^k
    (\alpha_j^\rT E(i\lambda)^*)
    \beta_k
    \lprod_{j=1}^k
    (E(i\lambda) \alpha_j)
    E(i\lambda)\\
\nonumber
    & =
    \sum_{k=0}^{+\infty}
    (-2i\theta)^k
    \phi_k
    G_k(i\lambda)^* \beta_k G_k(i\lambda)\\
\label{SigGH}
    & =
    G(i\lambda)^* H_\theta G(i\lambda),
    \qquad
    \lambda \in \mR,
\end{align}
where use is made of (\ref{EEEdual})--(\ref{H}). The factorisation (\ref{DelGH}) is now obtained by combining (\ref{SigGH}) with (\ref{Delta}).
\end{proof}

The spectral factorisation (\ref{DelGH}) involves the infinite cascade of classical linear systems forming the system $G$, which will be used in a shaping filter for an auxiliary Gaussian process.

\section{QEF growth rate representation in terms of classical Gaussian processes}
\label{sec:class}

We will now relate the quantum QEF growth rate (\ref{UpsDel}) to a similar functional for stationary $\mR^{\infty}$-valued zero-mean Gaussian random processes $\xi$, $\eta$  produced from independent standard Wiener processes $\omega_1$, $\omega_2$ in $\mR^n$
by an infinite-dimensional shaping filter (in operator form)
\begin{equation}
\label{GG}
    \zeta
    :=
    \begin{bmatrix}
      \xi\\
      \eta
    \end{bmatrix}
    =
    (I_2 \ox G)
    \cR(S)\rd  \omega,
    \qquad
    \omega:=
    \begin{bmatrix}
      \omega_1\\
      \omega_2
    \end{bmatrix}.
\end{equation}
Here, the system $G$ in  (\ref{G}) is used (under the condition (\ref{gamdet})) along with  the following representation of the square root $S$ from (\ref{BOBroot}):
\begin{equation}
\label{cR}
    \cR(S)
    :=
      \begin{bmatrix}
      \Re S & -\Im S\\
      \Im S & \Re S
    \end{bmatrix}
    =
    I_2 \ox \Re S - \bJ\ox \Im S,
\end{equation}
where the matrix $\bJ$ is given by (\ref{bJ}); see Fig.~\ref{fig:GG}. 
\begin{figure}[htbp]
\centering
\unitlength=1mm
\linethickness{0.5pt}
\begin{picture}(55.00,62.00)
\put(10,15){\framebox(10,10)[cc]{\small$\Im S$}}
\put(10,35){\framebox(10,10)[cc]{\small$\Re S$}}
\put(15,30){\vector(0,1){5}}
\put(15,30){\vector(0,-1){5}}
\put(5,30){\line(1,0){10}}
\put(0,30){\makebox(0,0)[cc]{{\small$\omega_1$}}}
\put(40,35){\framebox(10,10)[cc]{\small$-\Im S$}}
\put(40,15){\framebox(10,10)[cc]{\small$\Re S$}}
\put(45,30){\vector(0,1){5}}
\put(45,30){\vector(0,-1){5}}
\put(55,30){\line(-1,0){10}}
\put(60,30){\makebox(0,0)[cc]{{\small$\omega_2$}}}
\put(20,20){\vector(1,0){7.5}}
\put(40,20){\vector(-1,0){7.5}}
\put(30,17.5){\vector(0,-1){7.5}}
\put(30,20){\circle{5}}
\put(30,20){\makebox(0,0)[cc]{{$+$}}}
\put(30,40){\circle{5}}
\put(30,40){\makebox(0,0)[cc]{{$+$}}}
\put(20,40){\vector(1,0){7.5}}
\put(40,40){\vector(-1,0){7.5}}
\put(30,42.5){\vector(0,1){7.5}}
\put(25,55){\vector(-1,0){10}}
\put(11,55){\makebox(0,0)[cc]{{\small$\xi$}}}
\put(11,5){\makebox(0,0)[cc]{{\small$\eta$}}}
\put(25,5){\vector(-1,0){10}}
\put(25,0){\framebox(10,10)[cc]{\small$G$}}
\put(25,50){\framebox(10,10)[cc]{\small$G$}}
\end{picture}
\caption{An infinite-dimensional shaping filter which produces $\mR^\infty$-valued stationary Gaussian random processes $\xi$, $\eta$ from independent standard Wiener processes $\omega_1$, $\omega_2$ in $\mR^n$ and uses the square root $S$ from (\ref{BOBroot}) and the system $G$ from (\ref{G}).}
\label{fig:GG}
\end{figure}
Since (\ref{BJB}), (\ref{OmegaT}), (\ref{BOBroot}) imply that     $\cR(S)$ is a real symmetric matrix satisfying
\begin{equation}
\label{cRS2}
    \cR(S)^2
    =
    \cR(B\Omega^\rT B^\rT)
    =
    \cR(BB^\rT - i\mho)
    =
    \begin{bmatrix}
      BB^\rT & \mho\\
      -\mho & BB^\rT
    \end{bmatrix},
\end{equation}
the spectral density of the process $\zeta$ in (\ref{GG}) is factorised by the transfer function of the system $G$ as
\begin{equation}
\label{specden}
    (I_2 \ox G(i\lambda))
    \begin{bmatrix}
      BB^\rT  & \mho\\
      -\mho & BB^\rT
    \end{bmatrix}
    (I_2 \ox G(i\lambda)^*),
    \qquad
    \lambda \in \mR.
\end{equation}
In accordance with the structure of the system $G$ in (\ref{G}), the processes $\xi$, $\eta$ are split into $\mR^n$-valued subvectors as
\begin{equation}
\label{xieta}
    \xi
    :=
    \begin{bmatrix}
      \xi_0\\
      \xi_1\\
      \xi_2\\
      \vdots
    \end{bmatrix},
    \qquad
    \eta
    :=
    \begin{bmatrix}
      \eta_0\\
      \eta_1\\
      \eta_2\\
      \vdots
    \end{bmatrix}
\end{equation}
and are governed by an infinite cascade of classical Ito SDEs \cite{KS_1991} driven by $\omega$:
\begin{align}
\label{SDE0}
    \rd
    \begin{bmatrix}
      \xi_0 \\
      \eta_0
    \end{bmatrix}
    & =
    (I_2\ox A)
    \begin{bmatrix}
      \xi_0 \\
      \eta_0
    \end{bmatrix}
    \rd t
    +
    \cR(S)
    \rd \omega,\\
\label{SDEk}
    \rd
    \begin{bmatrix}
      \xi_k \\
      \eta_k
    \end{bmatrix}
    & =
    \Big(
    (I_2\ox A)
    \begin{bmatrix}
      \xi_k \\
      \eta_k
    \end{bmatrix}
    +
    (I_2\ox \alpha_k)
    \begin{bmatrix}
      \xi_{k-1} \\
      \eta_{k-1}
    \end{bmatrix}
    \Big)
    \rd t,
    \qquad
    k \> 1
\end{align}
(the time arguments are omitted for brevity).
Due to the absence of diffusion terms in (\ref{SDEk}),
$\xi_k$, $\eta_k$ have continuously differentiable sample paths for all $k\>1$.
The initial state $(\xi(0),\eta(0))$ is independent of the standard Wiener process $\omega$ and is distributed  according to the unique invariant Gaussian measure for these SDEs  (which is well defined since the matrix $A$ is Hurwitz).
Associated with the subvectors of $\xi$, $\eta$ in (\ref{xieta}) are the following $\mR^{4n}$-valued stationary zero-mean Gaussian processes:
\begin{equation}
\label{xietak}
    \zeta_k :=
    \begin{bmatrix}
      \xi_{2k}\\
      \eta_{2k}\\
      \xi_{2k+1}\\
      \eta_{2k+1}
    \end{bmatrix},
    \qquad
    k = 0,1,2,\ldots.
\end{equation}
We will also use the real and imaginary parts of the matrix $H_\theta$ in  (\ref{H}) which are also block diagonal matrices:
\begin{align}
\label{ReH}
    \Re H_\theta
    & =
    \diag(
        f_0(\theta),0,\
        f_1(\theta),0,\
        f_2(\theta),0,\ldots
        ),\\
\label{ImH}
    \Im H_\theta
    & =
    \diag(
        0,g_0(\theta),\
        0,g_1(\theta),\
        0,g_2(\theta),\ldots
        ).
\end{align}
In view of Lemma~\ref{lem:alfbetgam},
the nontrivial diagonal blocks of $\Re H_\theta$ are real positive definite symmetric matrices
\begin{equation}
\label{fk}
    f_k(\theta)
    :=
    (-2i\theta)^{2k}\phi_{2k} \beta_{2k}
    =
    (-4\theta^2)^k\phi_{2k} \beta_{2k}
    \succ 0,
\end{equation}
whereas the nontrivial diagonal blocks of $\Im H_\theta$ are real antisymmetric matrices:
\begin{equation}
\label{gk}
    g_k(\theta)
    :=
    \frac{1}{i}
    (-2i\theta)^{2k+1}\phi_{2k+1} \beta_{2k+1}
    =
    -2\theta(-4\theta^2)^k\phi_{2k+1} \beta_{2k+1}
\end{equation}
for all $k = 0,1,2,\ldots$. We assemble (\ref{fk}), (\ref{gk})
into real symmetric matrices of order $4n$:
\begin{equation}
\label{hk}
    h_k(\theta)
    :=
    \begin{bmatrix}
        I_2 \ox f_k(\theta) & 0\\
        0& -\bJ \ox g_k(\theta)
    \end{bmatrix},
    \qquad
    k = 0,1,2,\ldots,
\end{equation}
where the matrix $\bJ$ from (\ref{bJ}) is used.
In view of (\ref{ReH})--(\ref{gk}),  a combination of the processes (\ref{xietak}) with the matrices (\ref{hk}) leads to the identity
\begin{align}
\nonumber
    Q_\theta
    & :=
      \zeta^\rT
      \cR(H_\theta)
    \zeta\\
\nonumber
    & =
      \xi^\rT
      \Re H_\theta \xi
      +
      \eta^\rT
      \Re H_\theta \eta - 2 \xi^\rT \Im H_\theta \eta\\
\nonumber
    & =
    \sum_{k=0}^{+\infty}
    (\xi_{2k}^\rT
      f_k(\theta)\xi_{2k}
      +
      \eta_{2k}^\rT f_k(\theta)\eta_{2k}
      -
      2
    \xi_{2k+1}^\rT
      g_k(\theta)\eta_{2k+1})\\
\label{Q}
    &
    =
    \sum_{k=0}^{+\infty}
    \zeta_k^\rT h_k(\theta)\zeta_k,
\end{align}
with the map $\cR$ from (\ref{cR}) being applied to the matrix $H_\theta$ in (\ref{H}). The almost sure convergence of the random series on the right-hand side  of (\ref{Q}) can be  established 
by using the structure of the matrices (\ref{hk}) along with the invariant zero-mean  Gaussian measure for the process $\zeta$ whose covariance matrix
\begin{equation}
\label{cP}
    \cP
    :=
    \bM(\zeta\zeta^\rT)
    =
    (\bM(\zeta_j\zeta_k^\rT))_{j,k\> 0}\in \mS_\infty^+
\end{equation}
(with $\bM(\cdot)$ the classical expectation)
satisfies the ALE
\begin{equation}
\label{cPALE}
    \cA \cP + \cP \cA^\rT + \cB\cB^\rT = 0.
\end{equation}
Here,
\begin{equation}
\label{cAB}
    \cA
    :=
    \begin{bmatrix}
        I_2\ox A   & 0 & 0 &\ldots &\ldots \\
        I_2\ox \alpha_1  & I_2\ox A & 0 & \ldots  &\ldots \\
        0   & I_2\ox \alpha_2 & I_2\ox A & \ldots  &\ldots \\
        \ldots    &\ldots   & \ldots  & \ldots  &\ldots
    \end{bmatrix},
    \qquad
    \cB
    :=
    \begin{bmatrix}
      \cR(S)\\
      0\\
      0\\
      \vdots
    \end{bmatrix}
\end{equation}
are infinite-dimensional matrices formed from the state-space matrices of the SDEs (\ref{SDE0}), (\ref{SDEk}).
Therefore, (\ref{Q}) defines a strictly stationary real-valued random process $Q_\theta$ whose expectation
\begin{equation}
\label{MQ}
    \bM Q_\theta
    =
    \frac{1}{2\pi}
    \int_{\mR}
    \Tr \Pi_\theta(\lambda)
    \rd \lambda
\end{equation}
(at any moment of time)
 is expressed
in terms of another spectral function $\Pi_\theta: \mR\to \mH_{2n}$ given by
\begin{equation}
\label{Pi}
    \Pi_\theta(\lambda)
    :=
    \cR(S)
    (I_2 \ox G(i\lambda)^*)
    \cR(H_\theta)(I_2 \ox G(i\lambda))
    \cR(S),
    \qquad
    \lambda\in \mR,
\end{equation}
in accordance with the spectral density (\ref{specden}) of the process $\zeta$ in (\ref{GG}). The following lemma provides a link between $\Pi_\theta$ and the 
spectral density $\Delta_\theta$ in (\ref{Delta}).

\begin{lem}
\label{lem:Pi}
Under the condition (\ref{gamdet}), the function $\Pi_\theta$ in  (\ref{Pi}) is related to the spectral density $\Delta_\theta$ in (\ref{Delta}) by
\begin{equation}
\label{PiDel}
    \Pi_\theta(\lambda)
    =
    N
    \begin{bmatrix}
        \Delta_\theta(\lambda) & 0\\
        0 & \Delta_\theta(-\lambda)^\rT
    \end{bmatrix}
    N^*,
    \qquad
    \lambda \in \mR,
\end{equation}
where
\begin{equation}
\label{N}
  N:=
  \frac{1}{\sqrt{2}}
    \begin{bmatrix}
        1 & 1\\
        -i & i
    \end{bmatrix}
    \ox
    I_n
\end{equation}
is a unitary matrix. In particular, $\Pi_\theta$ takes values in $\mH_{2n}^+$. 
\hfill$\square$
\end{lem}
\begin{proof}
From (\ref{cR}), (\ref{N}), it follows that the matrix $N$ secures the unitary equivalence
\begin{equation}
\label{Nabc}
    \cR(c)
    =
    N
    \begin{bmatrix}
      c & 0\\
        0 & \overline{c}
    \end{bmatrix}
    N^*,
    \qquad
    c\in \mC^{n\x n}.
\end{equation}
Also, both $N$ and $N^*$, due to their Kronecker product structure,  commute with $I_2\ox c$:
\begin{equation}
\label{Nc}
    [N, I_2\ox c] = 0,
    \qquad
    [N^*,I_2\ox c] = 0,
    \qquad
    c \in \mC^{n\x n}.
\end{equation}
Repeated application of (\ref{Nabc}), (\ref{Nc}) to (\ref{Pi}) leads to
\begin{align}
\nonumber
    \Pi_\theta(\lambda)
    & =
    N
    \begin{bmatrix}
      S & 0\\
      0 & \overline{S}
    \end{bmatrix}
    N^*
    (I_2 \ox G(i\lambda)^*)
    \cR(H_\theta)
    (I_2 \ox G(i\lambda))
    N
    \begin{bmatrix}
      S & 0\\
      0 & \overline{S}
    \end{bmatrix}
    N^*\\
\nonumber
    & =
    N
    \begin{bmatrix}
      S & 0\\
      0 & \overline{S}
    \end{bmatrix}
    (I_2 \ox G(i\lambda)^*)
        N^*
        \cR(H_\theta)
    N    (I_2 \ox G(i\lambda))
    \begin{bmatrix}
      S & 0\\
      0 & \overline{S}
    \end{bmatrix}
    N^*    \\
\nonumber
    & =
    N
    \begin{bmatrix}
      S & 0\\
      0 & \overline{S}
    \end{bmatrix}
    (I_2 \ox G(i\lambda)^*)
    \begin{bmatrix}
      H_\theta & 0\\
      0 & \overline{H_\theta}
    \end{bmatrix}(I_2 \ox G(i\lambda))
    \begin{bmatrix}
      S & 0\\
      0 & \overline{S}
    \end{bmatrix}
    N^*\\
\nonumber
    & =
    N
    \begin{bmatrix}
      S G(i\lambda)^* H_\theta G(i\lambda) S  & 0\\
      0 & \overline{S} G(i\lambda)^* \overline{H_\theta}G(i\lambda)\overline{S}
    \end{bmatrix}
    N^*\\
\label{PiNN}
    & =
    N
    \begin{bmatrix}
      \Delta_\theta(\lambda)  & 0\\
      0 & \overline{S} G(i\lambda)^* \overline{H_\theta}G(i\lambda)\overline{S}
    \end{bmatrix}
    N^*,
    \qquad
    \lambda \in \mR,
\end{align}
where the last equality also uses the factorisation (\ref{DelGH}).
Since the matrices $S$,  $H_\theta$ in (\ref{BOBroot}), (\ref{H}) are Hermitian, then (\ref{DelGH}) implies that
\begin{align}
\nonumber
  \Delta_\theta(-\lambda)^\rT
  & =
  (SG(-i\lambda)^* H_\theta G(-i\lambda)S)^\rT\\
\nonumber
  & =
  S^\rT G(-i\lambda)^\rT H_\theta^\rT\, \overline{G(-i\lambda)}S^\rT\\
\label{DelGHT}
  & =
  \overline{S} G(i\lambda)^* \overline{H_\theta}G(i\lambda)\overline{S},
  \qquad
  \lambda \in \mR,
\end{align}
where use is also made of the relation $G(i\lambda)^* = G(-i\lambda)^\rT$   in view of the system $G$ in (\ref{G}) having real state-space matrices. Substitution of (\ref{DelGHT}) into (\ref{PiNN}) establishes (\ref{PiDel}). The latter implies that $\Pi_\theta(\lambda)\succcurlyeq 0$ since $\Delta_\theta(\lambda), \Delta_\theta(-\lambda)^\rT \succcurlyeq 0$ for any $\lambda\in \mR$.
\end{proof}

The unitary equivalence (\ref{PiDel}) allows the relation (\ref{MQ}) to be represented in terms of the spectral density $\Delta_\theta$ as 
$$
    \bM Q_\theta
    =
    \frac{1}{2\pi}
    \int_{\mR}
    \Tr (\Delta_\theta(\lambda) + \Delta_\theta(-\lambda))
    \rd \lambda
    =
    \frac{1}{\pi}
    \int_{\mR}
    \Tr \Delta_\theta(\lambda)
    \rd \lambda.
$$
The following theorem exploits a similar connection for the exponential-of-integral moments of the process $Q_\theta$.

\begin{thm}
\label{th:QEFclass}
Under the condition (\ref{gamdet}), for any $\theta>0$ satisfying (\ref{spec1}), the QEF growth rate (\ref{Ups}) is related by
\begin{equation}
\label{Upsclass}
    \Ups(\theta)
     =
    \frac{1}{2}
    \lim_{T\to +\infty}
    \Big(
    \frac{1}{T}
    \ln
    \bM
    \re^{
    \frac{\theta}{2}
    \int_0^T
    Q_\theta(t)
    \rd t}
    \Big)
\end{equation}
to the process $Q_\theta$ in (\ref{Q}). 
\hfill$\square$
\end{thm}
\begin{proof}
The process $Q_\theta$ in (\ref{Q}) is a quadratic form, with the matrix  (\ref{H}),  in the Gaussian random process $\zeta$ from (\ref{GG}) with the spectral density (\ref{specden}).
By using its truncation
\begin{equation}
\label{Qr}
  Q_{\theta,r}
  :=
      \sum_{k=0}^r
    \zeta_k^\rT
    h_k(\theta)
    \zeta_k,
\end{equation}
applying  the Fredholm determinant formula \cite[Theorem 3.10 on p. 36]{S_2005} (see also \cite{G_1994})  along with the asymptotic infinite-hori\-zon behaviour of spectra of Toeplitz operators, and passing to the limit as $r\to +\infty$ in (\ref{Qr}) combined with a uniform integrability argument, 
it follows that the classical QEF rate for $\zeta$ on the right-hand side of (\ref{Upsclass}) is related to the function $\Pi_\theta$ in (\ref{Pi}) as
\begin{equation}
\label{lndet1}
    \lim_{T\to +\infty}
    \Big(
        \frac{1}{T}
        \ln
        \bM
        \re^{
        \frac{\theta}{2}
        \int_0^T
        Q_\theta(t)
        \rd t}
    \Big)
    =
    -
    \frac{1}{4\pi}
    \int_{\mR}
    \ln\det(I_{2n}
    -
    \theta \Pi_\theta(\lambda)
    )
    \rd \lambda.
\end{equation}
The unitarity of the matrix $N$ in (\ref{N}) and the relation (\ref{PiDel}) imply the  unitary equivalence
$$
    I_{2n}-\theta \Pi_\theta(\lambda)
     =
    N
    \begin{bmatrix}
        I_n - \theta \Delta_\theta(\lambda) & 0\\
        0 & I_n - \theta \Delta_\theta(-\lambda)^\rT
    \end{bmatrix}
    N^*,
$$
whereby
$$
    \det(I_{2n}-\theta \Pi_\theta(\lambda))
    =
    \det (I_n - \theta \Delta_\theta(\lambda)) \det (I_n - \theta \Delta_\theta(-\lambda)),
$$
and hence,
\begin{align}
\nonumber
    \int_{\mR}
    \ln\det(I_{2n}
    -
    \theta \Pi_\theta(\lambda)
    )
    \rd \lambda
    & =
    \int_{\mR}
    (
    \ln\det(I_n - \theta \Delta_\theta(\lambda))
    +
        \ln\det(I_n - \theta \Delta_\theta(-\lambda)))
        \rd \lambda\\
\label{lndet2}
    & =
    2
    \int_{\mR}
    \ln\det(I_n - \theta \Delta_\theta(\lambda))
    \rd \lambda.
\end{align}
Substitution of (\ref{lndet2}) into (\ref{lndet1}) leads to
\begin{equation}
\label{lndet3}
    \lim_{T\to +\infty}
    \Big(
        \frac{1}{T}
        \ln
        \bM
        \re^{
        \frac{\theta}{2}
        \int_0^T
        Q_\theta(t)
        \rd t}
    \Big)
    =
    -
    \frac{1}{2\pi}
    \int_{\mR}
    \ln\det(I_n
    -
    \theta \Delta_\theta(\lambda)
    )
    \rd \lambda.
\end{equation}
The representation (\ref{Upsclass}) for the quantum QEF rate $\Ups(\theta)$ is now obtained by comparing (\ref{UpsDel}) of Lemma~\ref{lem:UpsDel} with (\ref{lndet3}).
\end{proof}

Theorem~\ref{th:QEFclass} reduces the computation of the QEF growth rate to that for the infinite-dimensional Gaussian process $\zeta$. The classical QEF rate on the right-hand side of (\ref{Upsclass}) can be found with arbitrary accuracy by using the truncation (\ref{Qr}). In view of (\ref{fk})--(\ref{hk}),  this corresponds to retaining only the first $2r+2$ terms in the Taylor series expansion of $\phi$ in (\ref{phi}) as
\begin{equation}
\label{phi2r1}
    \phi(u)\approx \sum_{k=0}^{2r+1} \phi_k u^k.
\end{equation}
Instead of $\phi_0, \ldots, \phi_{2r+1}$  in (\ref{fk}), (\ref{gk}),    alternative coefficients can also be used in order for the resulting approximation to retain qualitative properties of the function $\phi$   (such as positiveness) in addition to $\phi(0)=1$, which will be discussed in Section~\ref{sec:polrat}.

\section{Approximate QEF  rate computation using a truncated Gaussian process}
\label{sec:trunc}

Theorem~\ref{th:QEFclass} can be practically used by computing the classical QEF rate for an  $\mR^\nu$-valued Gaussian diffusion process
\begin{equation}
\label{Zr}
    Z_r
    :=
    \begin{bmatrix}
    \zeta_0\\
    \vdots\\
    \zeta_r
    \end{bmatrix}
\end{equation}
of dimension
\begin{equation}
\label{nu}
  \nu:= 4(r+1)n,
\end{equation}
where the parameter $r=0,1,2,\ldots$ controls the quality of approximating the exact value of the quantum QEF rate in (\ref{Upsclass}). The process $Z_r$ consists of $\zeta_0, \ldots, \zeta_r$ from (\ref{xietak}), in terms of which the truncation  (\ref{Qr}) of the process $Q_\theta$ in (\ref{Q}) is represented as
\begin{equation}
\label{QZr}
  Q_{\theta,r}
  :=
  Z_r^\rT
  H_{\theta,r}
  Z_r,
\end{equation}
where
\begin{equation}
\label{Hr}
    H_{\theta,r}
    :=
    \diag_{0\< k\< r}
    h_k(\theta) \in \mS_\nu
\end{equation}
is the corresponding submatrix of the matrix $\cR(H_\theta)$ associated with (\ref{H}), (\ref{hk}). In view of the SDEs (\ref{SDE0}), (\ref{SDEk}), the process $Z_r$ in (\ref{Zr}) is produced from the $\mR^{2n}$-valued standard Wiener process $\omega$ by a finite-dimensional shaping filter with the state-space realisation
\begin{equation}
\label{cABI}
    \left[
    \begin{array}{c|c}
    \cA_r & \cB_r\\
      \hline
      I_\nu &  0
    \end{array}
    \right],
\end{equation}
where $\cA_r \in \mR^{\nu\x \nu}$, $\cB_r\in \mR^{\nu\x 2n}$ are submatrices of $\cA$, $\cB$ from  (\ref{cAB}) given by
\begin{equation}
\label{cABr}
    \cA_r
    :=
    \begin{bmatrix}
        I_2\ox A   & 0 & 0 &\ldots &\ldots & 0\\
        I_2\ox \alpha_1  & I_2\ox A & 0 & \ldots  &\ldots & 0\\
        0   & I_2\ox \alpha_2 & I_2\ox A & \ldots  &\ldots & 0\\
        \ldots    &\ldots   & \ldots  & \ldots  &\ldots &  \ldots\\
        \ldots    &\ldots   & \ldots  & \ldots  &I_2\ox A &  0\\
        \ldots    &\ldots   & \ldots  & \ldots  &I_2\ox \alpha_{2r+1} &  I_2\ox A
    \end{bmatrix},
    \quad
    \cB_r
    :=
    \begin{bmatrix}
      \cR(S)\\
      0\\
      \vdots\\
      0
    \end{bmatrix},
\end{equation}
with $\cA_r$ inheriting the Hurwitz property from $A$. Note that the matrices $\cA_r$ and $H_{\theta,r}$ in (\ref{Hr})  involve the matrices  $\alpha_1, \ldots, \alpha_{2r+1}$    and $\beta_0, \ldots, \beta_{2r+1}$ from (\ref{alfbetgamnext})--(\ref{bet0}). In order for these matrices  (and hence, the related processes $Z_r$, $Q_{\theta,r}$ in (\ref{QZr}),  (\ref{Zr})) to be well-defined, the condition (\ref{gamdet}) can be replaced with its weaker version
\begin{equation}
\label{gamdetr}
  \det \gamma_j \ne 0,
  \qquad
  j =0, \ldots, 2r,
\end{equation}
where $\det \gamma_0\ne 0$ holds in view of (\ref{alfbetgam0}) due to the nonsingularity of the CCR matrix $\Theta$.
The following theorem (its proof is outlined below for completeness) applies the results of  \cite[Proposition 6.3.1 and its proof on pp. 66--68]{MG_1990} (see also \cite{BV_1985}) to computing the QEF rate for the truncated processes. The slight modification here is that, in contrast to  the standard risk-sensitive settings,  the matrix $H_{\theta,r}$ in (\ref{Hr}) is indefinite. 

\begin{thm}
\label{th:QEFr}
For a given $r\> 0$,  suppose (\ref{gamdetr}) holds for  the matrices $\gamma_0, \ldots, \gamma_{2r}$ in (\ref{alfbetgamnext}), (\ref{alfbetgam0}). Also, suppose the  risk sensitivity parameter $\theta>0$ and the matrix $H_{\theta,r}$ in (\ref{Hr}) satisfy
\begin{equation}
\label{spec2}
    \theta
    \sup_{\lambda\in \mR}
    \lambda_{\max}(f_r(i\lambda)^*H_{\theta,r} f_r(i\lambda))
    < 1,
    \qquad
    f_r(s):= (sI_\nu-\cA_r)^{-1}\cB_r,
\end{equation}
where $f_r$  is
the transfer function of the shaping filter (\ref{cABI}) of the stationary Gaussian diffusion process $Z_r$ with the matrices (\ref{cABr}). Then the process $Q_{\theta,r}$ in (\ref{Qr}), associated with $Z_r$  by (\ref{QZr}), satisfies
\begin{equation}
\label{QEFrate}
    \lim_{T\to +\infty}
    \Big(
        \frac{1}{T}
        \ln
        \bM
        \re^{\frac{\theta}{2} \int_0^T Q_{\theta,r}(t)\rd t}
    \Big)
    =
    \frac{1}{2}
    \Tr
    (
    \cB_r\cB_r^\rT    a_{\theta,r}),
\end{equation}
where $a_{\theta,r} \in \mS_\nu$ is the unique stabilising solution of the ARE
\begin{equation}
\label{AREr}
    \cA_r^\rT a_{\theta,r} + a_{\theta,r}\cA_r
   + \theta H_{\theta,r} +
  a_{\theta,r}\cB_r\cB_r^\rT a_{\theta,r} = 0
\end{equation}
in the sense that the matrix $\cA_r + \cB_r\cB_r^\rT a_{\theta,r} $ is Hurwitz.
\hfill$\square$
\end{thm}
\begin{proof}
As mentioned above, (\ref{gamdetr}) makes the processes $Z_r$ in (\ref{Zr})  and $Q_{\theta,r}$ in (\ref{QZr}) well-defined.
Omitting the subscripts $\theta$, $r$ for brevity, the Gaussian diffusion process $Z$ is governed by the SDE
\begin{equation}
\label{dZ}
  \rd Z = \cA Z\rd t + \cB \rd \omega,
\end{equation}
which describes the shaping filter (\ref{cABI}) driven by the standard Wiener  process $\omega$ in $\mR^{2n}$, with the matrices $\cA$, $\cB$ from (\ref{cABr}).  Assuming the risk sensitivity parameter $\theta$ to be fixed, the conditional QEF
\begin{equation}
\label{QEFT}
    K_T(z)
    :=
    \bM
    \big(
    \re^{\frac{\theta}{2} \int_0^T Z(t)^\rT H Z(t)\rd t}\,
    \big|\,
    Z(0)=z
    \big)>0,
    \qquad
    T \> 0, \
    z \in \mR^\nu,
\end{equation}
for the process $Z$
over a finite time horizon $T$ satisfies the integro-differential equation
\begin{equation}
\label{KIDE}
    K_{T+\tau}(z)
    =
    \bM
    \big(
    K_T(Z(\tau))
    \re^{\frac{\theta}{2} \int_0^\tau Z(t)^\rT H Z(t)\rd t}\,
    \big|\,
    Z(0)=z
    \big),
    \qquad
    \tau \> 0,
\end{equation}
with the initial condition
\begin{equation}
\label{K0}
    K_0=1.
\end{equation}
This follows from
the homogeneous Markov property of $Z$ and the tower property of iterated conditional expectations.  By letting $\tau\to 0+$ and using the Ito lemma \cite{KS_1991}, (\ref{KIDE}) gives rise to
the PDE
\begin{equation}
\label{KPDE}
  \d_T K_T
  =
  \cG(K_T) + \frac{\theta}{2} z^\rT H z K_T,
\end{equation}
where $\cG$ is the infinitesimal generator of the diffusion process $Z$ in (\ref{dZ}) acting on a twice continuously differentiable function $\varphi: \mR^\nu\to  \mR$ with the gradient vector $\varphi'$ and the Hessian matrix $\varphi''$  as 
$$
    \cG(\varphi)(z)
    =
    z^\rT \cA^\rT \varphi'(z)  + \frac{1}{2} \Tr (\cB\cB^\rT \varphi''(z)),
    \qquad
    z \in \mR^\nu.
$$
Since $     \frac{1}{\varphi}
    \cG(\varphi)
    =
    \cG(\ln \varphi)
    +
    \frac{1}{2}
    |\cB^\rT(\ln \varphi)'|^2
$ for positive functions $\varphi$
in accordance with the
Fleming logarithmic transformation \cite{F_1982}
(see also \cite[Eq.~(81) on p.~201]{BFP_2002}), then (\ref{KPDE}) takes the form
\begin{equation}
\label{KPDElog}
  \d_T \ln K_T(z)
  =
  \cG(\ln K_T)(z)
  +
  \frac{1}{2}
    |\cB^\rT(\ln K_T)'(z)|^2 + \frac{\theta}{2} z^\rT H z,
\end{equation}
which admits a quadratic  ansatz
\begin{equation}
\label{Klog}
    \ln K_T(z)
    =
    \frac{1}{2} z^\rT a_T z
    + c_T,
\end{equation}
where $a_T$, 
$c_T$ are continuously differentiable functions of $T$ with values in $\mS_\nu$, 
$\mR$,  respectively, and zero initial conditions
\begin{equation}
\label{a0c0}
    a_0=0,
    \qquad
    c_0 = 0
\end{equation}
in view of (\ref{K0}).   By substituting (\ref{Klog}) into (\ref{KPDElog}) and equating the corresponding coefficients of the quadratic functions in
$$
    z^\rT \dot{a}_T z +
    2\dot{c}_T
   =
  2z^\rT \cA^\rT
  a_T z
    + \Tr (\cB\cB^\rT a_T)+
    |\cB^\rT
    a_T z
    |^2  +
    \theta z^\rT H z
$$
(where $\dot{(\ )}:= \d_T(\cdot)$ is the time derivative), it follows that
\begin{align}
\label{adot}
    \dot{a}_T
  & =
  \cA^\rT a_T + a_T\cA
   + \theta H +
  a_T \cB\cB^\rT a_T,\\
\label{cdot}
    \dot{c}_T
  & =
  \frac{1}{2}
  \Tr (\cB\cB^\rT a_T).
\end{align}
The solution of the Riccati ODE (\ref{adot}) with the zero initial condition in (\ref{a0c0}) has a limit
\begin{equation}
\label{alim}
    a_\infty:= \lim_{T\to +\infty} a_T
\end{equation}
which is the unique stabilising (in the sense that  $\cA + \cB\cB^\rT a_\infty$ is Hurwitz) solution  of the ARE
$$
    \cA^\rT a_\infty + a_\infty\cA
   + \theta H +
  a_\infty \cB\cB^\rT a_\infty = 0
$$
in (\ref{AREr}). Therefore,
integration of (\ref{cdot}) with the zero initial condition from (\ref{a0c0}) yields
\begin{equation}
\label{clim}
    \lim_{T\to +\infty}
    \frac{c_T}{T}
    =
    \frac{1}{2}
    \Tr
    \Big(
    \cB\cB^\rT
    \lim_{T\to +\infty}
    \Big(
    \frac{1}{T}
    \int_0^T
    a_t \rd t
    \Big)
    \Big)
    =
    \frac{1}{2}
    \Tr
    (
    \cB\cB^\rT    a_\infty)
\end{equation}
since the Cesaro mean inherits the limit value (\ref{alim}). It now remains to note that
\begin{equation}
\label{Mexp0}
    \bM
    \re^{\frac{\theta}{2} \int_0^T Z(t)^\rT H Z(t)\rd t}
    =
    \bM K_T (Z(0))
    =
    \re^{c_T}
    \bM\re^{\frac{1}{2} Z(0)^\rT a_T Z(0)}
\end{equation}
in view of (\ref{QEFT}), (\ref{Klog}), where the rightmost expectation is  over the invariant zero-mean Gaussian distribution of the process $Z$ with an appropriate submatrix $\cP$ of the covariance matrix in (\ref{cP}), (\ref{cPALE}).
The condition (\ref{spec2}) implies that
$
    \lambda_{\max}(a_\infty \cP) < 1
$, which secures  a finite limit for the rightmost expectation in (\ref{Mexp0}):
$$
    \lim_{T\to +\infty}
    \bM\re^{\frac{1}{2} Z(0)^\rT a_T Z(0)}
    =
    \frac{1}{\sqrt{\det(I_\nu - a_\infty\cP)}} < +\infty.
$$
Hence, (\ref{clim}) implies that
$$
    \lim_{T\to +\infty}
    \Big(
        \frac{1}{T}
        \ln
        \bM
        \re^{\frac{\theta}{2} \int_0^T Z(t)^\rT H Z(t)\rd t}
    \Big)
    =
    \lim_{T\to +\infty}
    \frac{c_T}{T}
    =
    \frac{1}{2}
    \Tr
    (
    \cB\cB^\rT    a_\infty),
$$
thus establishing (\ref{QEFrate}).
\end{proof}

A combination of Theorems~\ref{th:QEFclass}, \ref{th:QEFr} allows the quantum QEF growth rate (\ref{Ups}) to be computed as the limit
\begin{equation}
\label{Upslim}
    \Ups(\theta)
    =
    \lim_{r\to +\infty}
    \Ups_r(\theta),
    \qquad
    \Ups_r(\theta)
    :=
    \frac{1}{4}
    \Tr
    (
    \cB_r\cB_r^\rT    a_{\theta,r})
\end{equation}
in terms of the ``truncated'' classical QEF rates  (\ref{QEFrate}).    A practical application of (\ref{Upslim}) consists in using the ``prelimit'' value $\Ups_r(\theta)$ at a finite $r$ large enough for the convergence to manifest itself. In this regard, of interest is the question of exploiting the special structure of the matrices $\cA_r$, $\cB_r$ in (\ref{cABr})  and $H_{\theta, r}$ in (\ref{Hr}) for computing the truncated QEF rate (\ref{QEFrate}) recursively in $r$. To this end, we note that the sparsity of $\cB_r$ in (\ref{cABr}) leads to
\begin{equation}
\label{cBB}
  \cB_r\cB_r^\rT
  =
  \begin{bmatrix}
    \cR(S)^2 & 0\\
    0 & 0
  \end{bmatrix}
  =
  \begin{bmatrix}
    \cR(B\Omega^\rT B^\rT) & 0\\
    0 & 0
  \end{bmatrix}  ,
\end{equation}
where the block $\cR(S)^2$ is computed in (\ref{cRS2}). Hence, only the first diagonal block $(a_{\theta,r})_{11} \in \mS_{2n}$ of the matrix $a_{\theta,r}$ enters the trace in (\ref{QEFrate}):
\begin{equation}
\label{a11}
    \Tr
    (\cB_r\cB_r^\rT    a_{\theta,r})
    =
    \Tr(\cR(B\Omega^\rT B^\rT) (a_{\theta,r})_{11}).
\end{equation}
Also, due to (\ref{cBB}), the ARE (\ref{AREr}) has a ``low-rank''  nonlinearity in the sense that its quadratic term 
depends only on the first block-column $(a_{\theta,r})_{\bullet 1} \in \mR^{\nu \x 2n}$ of the matrix $a_{\theta,r}$:
\begin{equation}
\label{aBBa}
    a_{\theta,r}\cB_r\cB_r^\rT a_{\theta,r}
    =
    (a_{\theta,r})_{\bullet 1}
    \cR(B\Omega^\rT B^\rT)
    (a_{\theta,r})_{\bullet 1}^\rT.
\end{equation}
Another relevant observation is that the solution $a_{\theta,r}$ of (\ref{AREr}) satisfies the ARE
\begin{equation}
\label{AREr1}
    \cA_r a_{\theta,r}^{-1} + a_{\theta,r}^{-1}\cA_r^\rT
   +
  \cB_r\cB_r^\rT + \theta a_{\theta,r}^{-1} H_{\theta,r}a_{\theta,r}^{-1} = 0
\end{equation}
(provided $\det a_{\theta,r}\ne 0$)
whose analysis can benefit from the block lower triangular structure of the matrix $\cA_r$ in (\ref{cABr}), similarly to that of $\cA$ in the ALE (\ref{cPALE}). This can also be combined with the Schur complement relations \cite{HJ_2007} between the blocks of $a_{\theta,r}$ and $a_{\theta,r}^{-1}$  in (\ref{a11})--(\ref{AREr1}). On the other hand, the factorially fast decay of the coefficients  (\ref{phik0}) leads to a rapid convergence in (\ref{Upslim}), so that already the initial approximation $\Ups_0(\theta)$ of the QEF rate appears to be satisfactory for moderate values of $\theta$, at least as the numerical example in Section~\ref{sec:exp} demonstrates. In view of (\ref{a11}), this approximation is described by
\begin{equation}
\label{Ups0}
    \Ups_0(\theta)
    =
    \frac{1}{4}
    \Tr(\cR(B\Omega^\rT B^\rT) (a_{\theta,0})_{11})
\end{equation}
in terms of the stabilising solution $a_{\theta,0}$  of the ARE (\ref{AREr}) of order $\nu=4n$ in (\ref{nu}) at $r=0$:
\begin{equation}
\label{ARE0}
    \cA_0^\rT a_{\theta,0} + a_{\theta,0}\cA_0
   + \theta h_0(\theta)+
  a_{\theta,0}\cB_0\cB_0^\rT a_{\theta,0} = 0.
\end{equation}
Here,
\begin{equation}
\label{cAB0}
    \cA_0
    =
    \begin{bmatrix}
        I_2\ox A   & 0 \\
        I_2\ox \Theta   & I_2\ox A
    \end{bmatrix},
    \qquad
    \cB_0
    =
    \begin{bmatrix}
      \cR(S)\\
      0
    \end{bmatrix}
\end{equation}
in view of (\ref{cABr}), (\ref{alfbetgam0}), and
\begin{equation}
\label{h0}
    h_0(\theta)
    =
    \begin{bmatrix}
        I_{2n} & 0\\
        0& \theta \bJ \ox (\Theta^{-1}\mho \Theta^{-1})
    \end{bmatrix}
\end{equation}
in accordance with (\ref{phik0}), (\ref{bet0}), (\ref{fk})--(\ref{hk}).  This approximation does not involve the ALEs (\ref{alfbetgamnext}) and reduces to solving the ARE (\ref{ARE0}).

\section{Square root polynomial approximation of the function $\phi$}
\label{sec:polrat}

The function $\phi$ from (\ref{phi})  enters the truncated QEF rate (\ref{QEFrate}) through its coefficients $\phi_0, \ldots, \phi_{2r+1}$ which participate in the matrix $H_{\theta,r}$ in (\ref{Hr}) in view of (\ref{fk})--(\ref{hk}). This corresponds to the approximation of $\phi$ by the appropriately truncated Taylor series in (\ref{phi2r1}). In contrast to $\phi$,  the resulting polynomial of degree $2r+1$ takes both positive and negative values. The following alternative approximation is nonnegative everywhere on the real axis, thus better reflecting the positiveness of the original function.  
Consider the Taylor series expansion for the square root of $\phi$ in (\ref{phi}). Denote its regular branch, with positive values on the real axis,  by
\begin{equation}
\label{psi}
  \psi(u):= \sqrt{\phi(u)} = \sum_{k=0}^{+\infty} \psi_k u^k,
  \qquad
  u \in \mC,
\end{equation}
where $\psi _k \in \mR$ are coefficients. The squaring of (\ref{psi}) relates the coefficients (\ref{phik0}) of (\ref{phi}) to $\psi _0, \psi _1, \psi _2, \ldots$ by the convolutions
\begin{equation}
\label{phik}
    \phi_k = \sum_{j=0}^k \psi _j \psi _{k-j},
    \qquad
    k = 0, 1, 2,\ldots.
\end{equation}
Since $\phi_0=1$, the relation (\ref{phik}) leads to a recurrence  equation for the coefficients of (\ref{psi}):
\begin{equation}
\label{psik}
    \psi _0 = 1,
    \qquad
    \psi _k
    =
    \frac{1}{2}
    \Big(
        \phi_k - \sum_{j=1}^{k-1} \psi _j \psi _{k-j}
    \Big),
    \qquad
    k = 1, 2,3, \ldots.
\end{equation}
In particular,
$$
    \psi _1 = \frac{1}{4},
    \qquad
    \psi _2 = \frac{5}{96},
    \qquad
    \psi _3 = \frac{1}{128},
    \qquad
    \psi _4 = 8.5720\x 10^{-4}.
$$
The numerical computation of the subsequent coefficients shows that they form a fast decaying sequence; see Fig.~\ref{fig:gdecay}. 
\begin{figure}[htbp]
\begin{center}
\includegraphics[width=12cm]{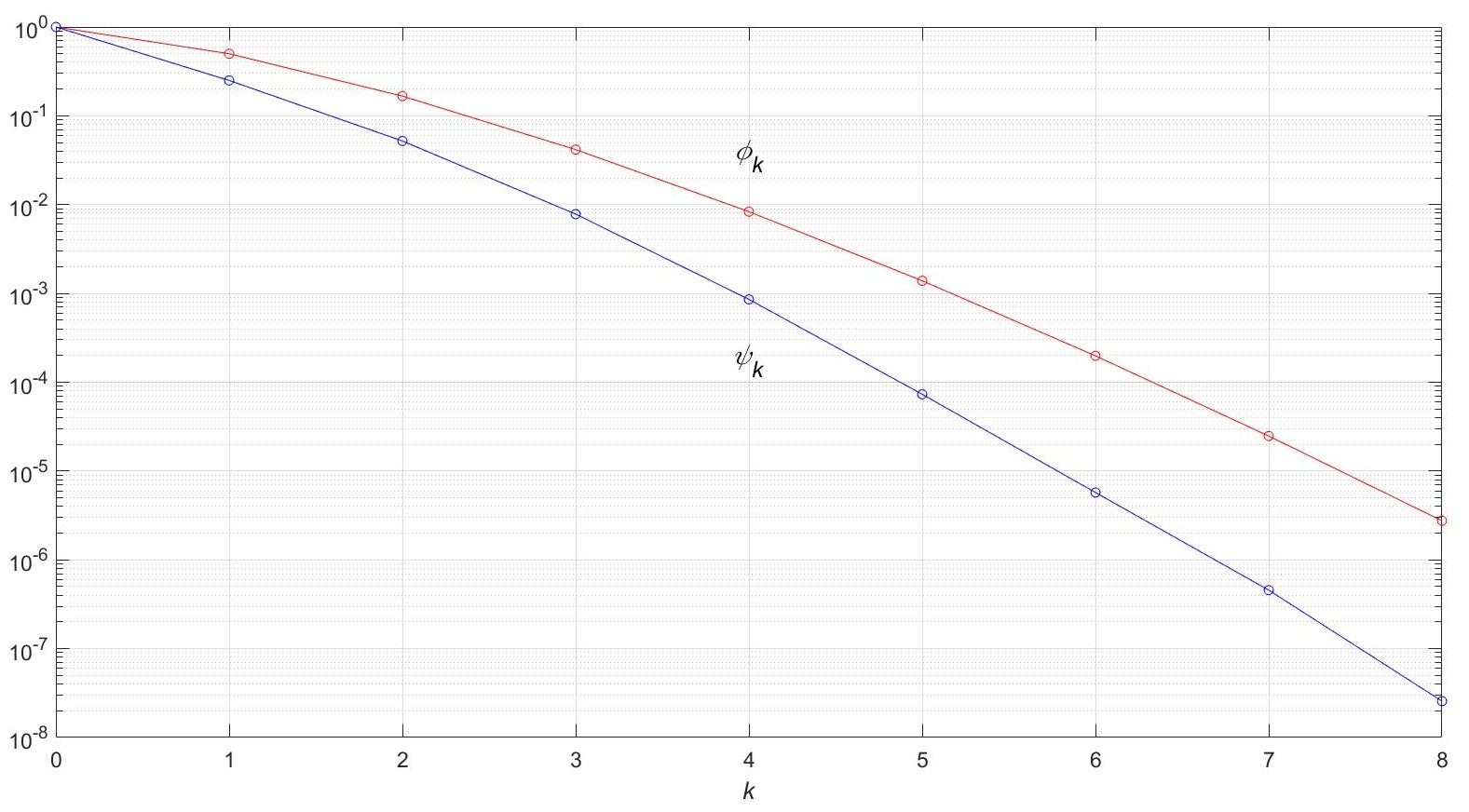}
\caption{The first nine coefficients $\phi_k$ from (\ref{phik0}) (red line) and $\psi _k$ from (\ref{psik}) (blue line)  on the logarithmic scale.}
\label{fig:gdecay}
\end{center}
\end{figure}
Therefore, the function $\phi$ in (\ref{phi}) admits a ``square root polynomial'' approximation (resembling the Choleski factorization \cite{HJ_2007} or the square root form of the Kalman filter \cite{B_1977}),
\begin{equation}
\label{phigr}
    \phi(u)
    \approx
    \Big(\sum_{k=0}^r \psi_k u^k\Big)^2
    =
    \sum_{k=0}^r
    \phi_k u^k
    +
    \sum_{k=r+1}^{2r}
    \sum_{j=k-r}^r
    \psi_j \psi_{k-j}
        u^k,
\end{equation}
with a reliable accuracy already for moderate values of the truncation parameter $r$, and, in contrast to (\ref{phi2r1}), being nonnegative on the real axis $u\in \mR$. The right-hand side of (\ref{phigr}) reproduces the first $r+1$ terms of the Taylor series expansion of $\phi$ in (\ref{phi}), while the coefficients $\phi_k$, with $r<k\< 2r$,  are replaced with     $\sum_{j=k-r}^r
    \psi_j \psi_{k-j}$, and the remainder $\sum_{k>2r}\phi_k u^k$ of the series is discarded.
These approximations are shown in Fig.~\ref{fig:phigr}. 
\begin{figure}[htbp]
\begin{center}
\includegraphics[width=12cm]{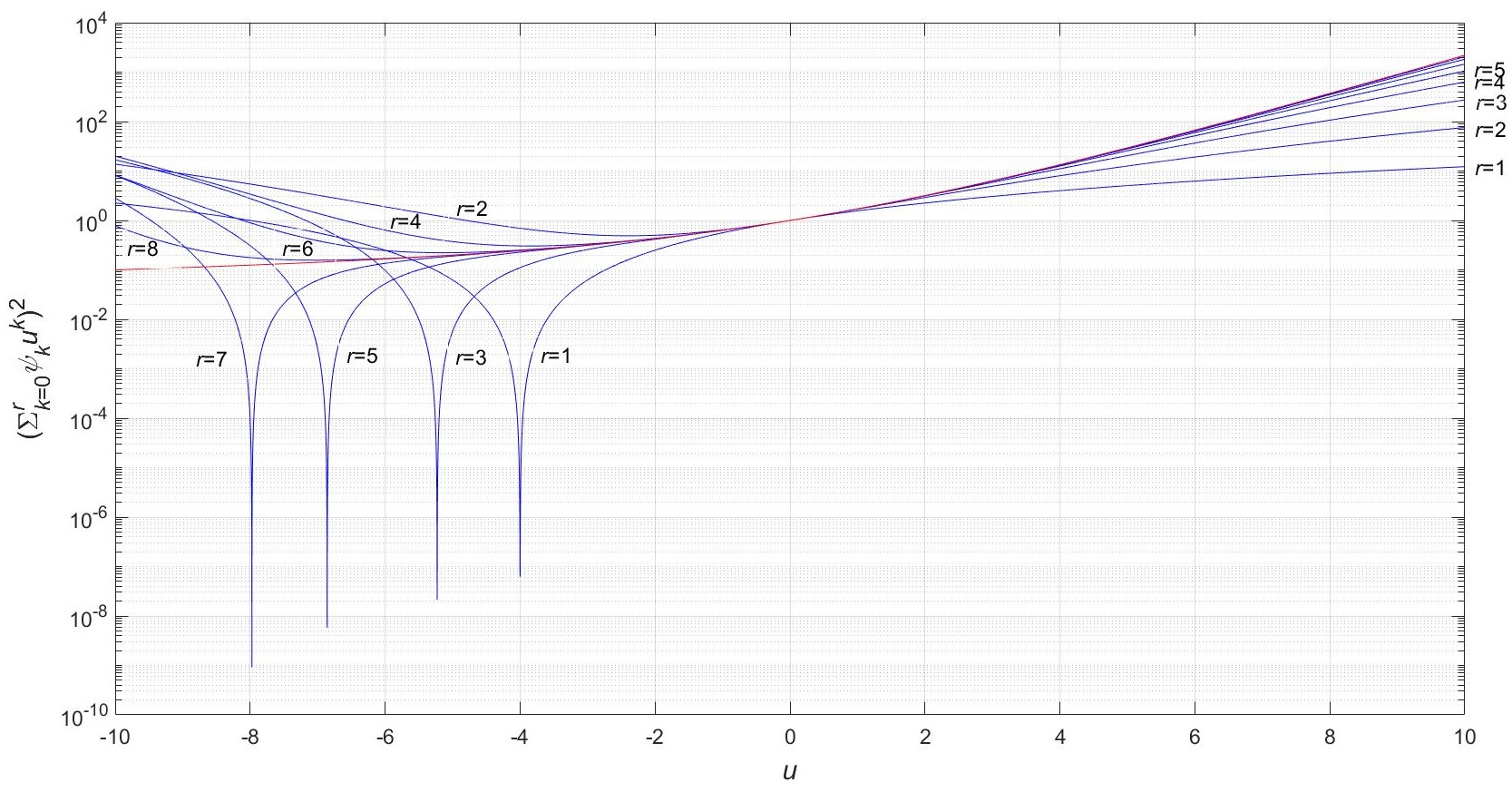}
\caption{The square root polynomial approximations (\ref{phigr}), with $r = 1, \ldots, 8$ (blue lines) for  the function (\ref{phi})  (red line) on the logarithmic scale. The troughs correspond to the real zeros of the polynomials $\sum_{k=0}^r \psi_k u^k$ for odd values of $r$.  The largest of these roots is $-4$ and pertains to the polynomial $1+\frac{1}{4}u$ with $r=1$.}
\label{fig:phigr}
\end{center}
\end{figure}

\section{A numerical example of computing the QEF rate in state space}
\label{sec:exp}

Two-mode OQHOs considered below, with $n=4$ system variables consisting of two position-momentum pairs \cite{S_1994}
and the CCR matrix  $\Theta= \frac{1}{2}\bJ\ox I_2$ in (\ref{XCCR}), can result as closed-loop systems from coherent (measurement-free) quantum feedback connections  of two cavities modelled by one-mode OQHOs. Mean-square optimal control settings for such systems, including applications to quantum optics, can be found, for example,  in \cite{NJP_2009}. The weighting of the system variables as in (\ref{SX})  (which yields  linear combinations of the positions and momenta in this case),   accompanied by  the transformations (\ref{STRM}), (\ref{SABC}),  can lead to a more complicated CCR matrix $\Theta$. For numerical illustration purposes, we will use the matrices
\begin{align*}
\Theta
    & :=
   {\scriptsize\begin{bmatrix}
         0 &   0.8697 &  -0.2444 &   0.4872\\
   -0.8697 &   0.0000 &   0.2612 &  -2.0179\\
    0.2444 &  -0.2612 &  -0.0000 &   1.1388\\
   -0.4872 &   2.0179 &  -1.1388 &  -0.0000
   \end{bmatrix}}, \\
A & :=
   {\scriptsize\begin{bmatrix}
    5.1304  &  5.7928 &   8.7655 &  -1.5445\\
   -9.0634  & -9.0965 & -14.7367 &   0.6865\\
    0.6371  &  0.1820 &  -0.6069 &   0.4491\\
   13.5996  &  5.4816 &   5.8039 &  -3.6400
   \end{bmatrix}},\\
B & :=
   {\scriptsize\begin{bmatrix}
    3.0301 &   0.1179 &   1.9804 &  -0.8723 &   0.9541 &   0.9578\\
   -2.1858 &  -2.0488 &  -2.0902 &   1.0467 &   1.7361 &  -0.2561\\
   -1.8423 &   0.7662 &  -0.4926 &   0.1977 &  -0.6104 &  -0.6082\\
   -3.3994 &  -3.6233 &  -2.0884 &  -1.5410 &   3.6763 &  -0.3150
   \end{bmatrix}},
\end{align*}
related by the PR condition (\ref{PR1}) and corresponding to the case of a three-mode ($m=6$) input quantum field with $J =\bJ\ox I_3$ in accordance with (\ref{JJ}). 
The matrix $A$ is Hurwitz, and its spectrum is $\{-1.3480 \pm 3.3108i, -2.7584 \pm 1.1650i\}$. Both $\Theta$ and $\mho$ in (\ref{BJB}) are nonsingular. The threshold value  (\ref{theta*}) of the risk sensitivity parameter $\theta$ and its classical counterpart (\ref{theta0}) are $\theta_* = 0.0792$ and  $\theta_0 = 0.0788$.
The results of numerical state-space computation of the QEF growth rate approximations $\Ups_r$,  based on the corollary (\ref{Upslim}) of Theorems~\ref{th:QEFclass}, \ref{th:QEFr},  are shown  for $r = 0, 1, 2, 3$ in Fig.~\ref{fig:Ups1} 
\begin{figure}[htbp]
\begin{center}
\includegraphics[width=12cm]{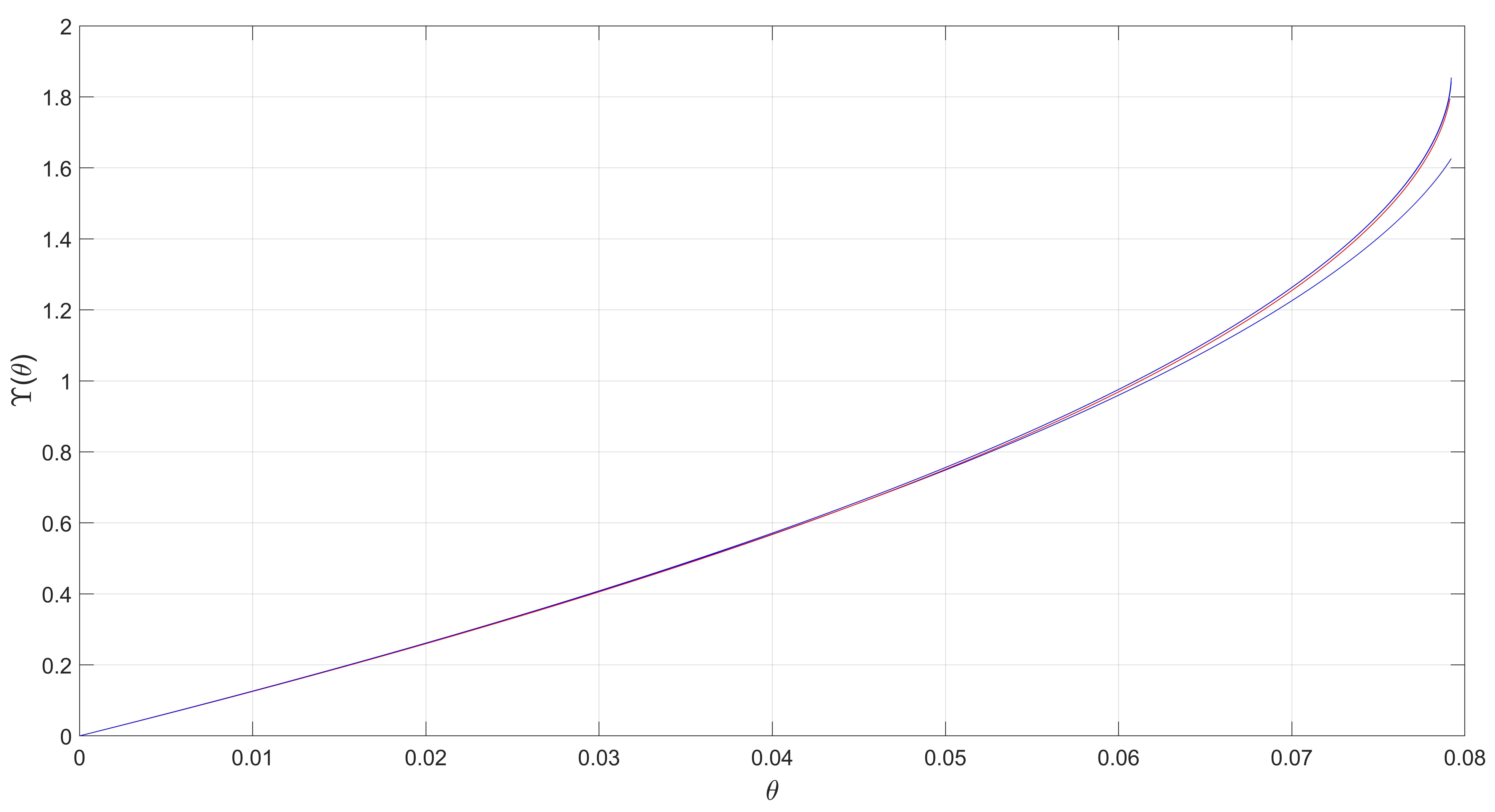}
\caption{The blue curves are the graphs of the QEF growth rate approximations $\Ups_r(\theta)$  (as functions of the risk sensitivity parameter $0\< \theta < \theta_*$) from (\ref{Upslim}) computed using Theorem~\ref{th:QEFr} for $r=0,1,2,3$ (the lower curve corresponds to $\Ups_0$ in (\ref{Ups0})--(\ref{h0})). The red curve is the outcome of the homotopy algorithm in the frequency domain \cite{VPJ_2020_IFAC,VPJ_2021}. }
\label{fig:Ups1}
\end{center}
\end{figure}
in comparison with those obtained through the homotopy algorithm of \cite{VPJ_2020_IFAC,VPJ_2021} in the frequency domain. No violation of the condition (\ref{gamdetr}) was observed.
The initial approximation $\Ups_0$, computed according to (\ref{Ups0})--(\ref{h0}), is satisfactory over the overwhelming part of the interval $[0,\theta_*)$, while the subsequent approximations $\Ups_1, \Ups_2, \Ups_3$ are practically indistinguishable from the frequency-domain result over the whole range.
The deviation of the graphs is noticeable only in the vicinity of the threshold value $\theta_*$. The four approximations $\Ups_0(\theta), \ldots, \Ups_3(\theta)$  are fairly close to each other even at a near-critical value of $\theta$ given in Tab.~\ref{tab:Ups}.
\begin{table}[htbp]
\caption{The values of the four QEF rate approximations at $\theta = 0.9999 \theta_* = 0.0792$.}
\begin{center}
{\scriptsize\begin{tabular}{|c||c|c|c|c|}
\hline
 $r$ & 0 & 1 & 2 & 3 \\
 \hline
 \hline
 $\Ups_r(\theta)$ &  1.6260 & 1.8427 & 1.8542 & 1.8543\\
\hline
\end{tabular}}
\label{tab:Ups}
\end{center}
\end{table}
Fig.~\ref{fig:Ups2} 
\begin{figure}[htbp]
\begin{center}
\includegraphics[width=12cm]{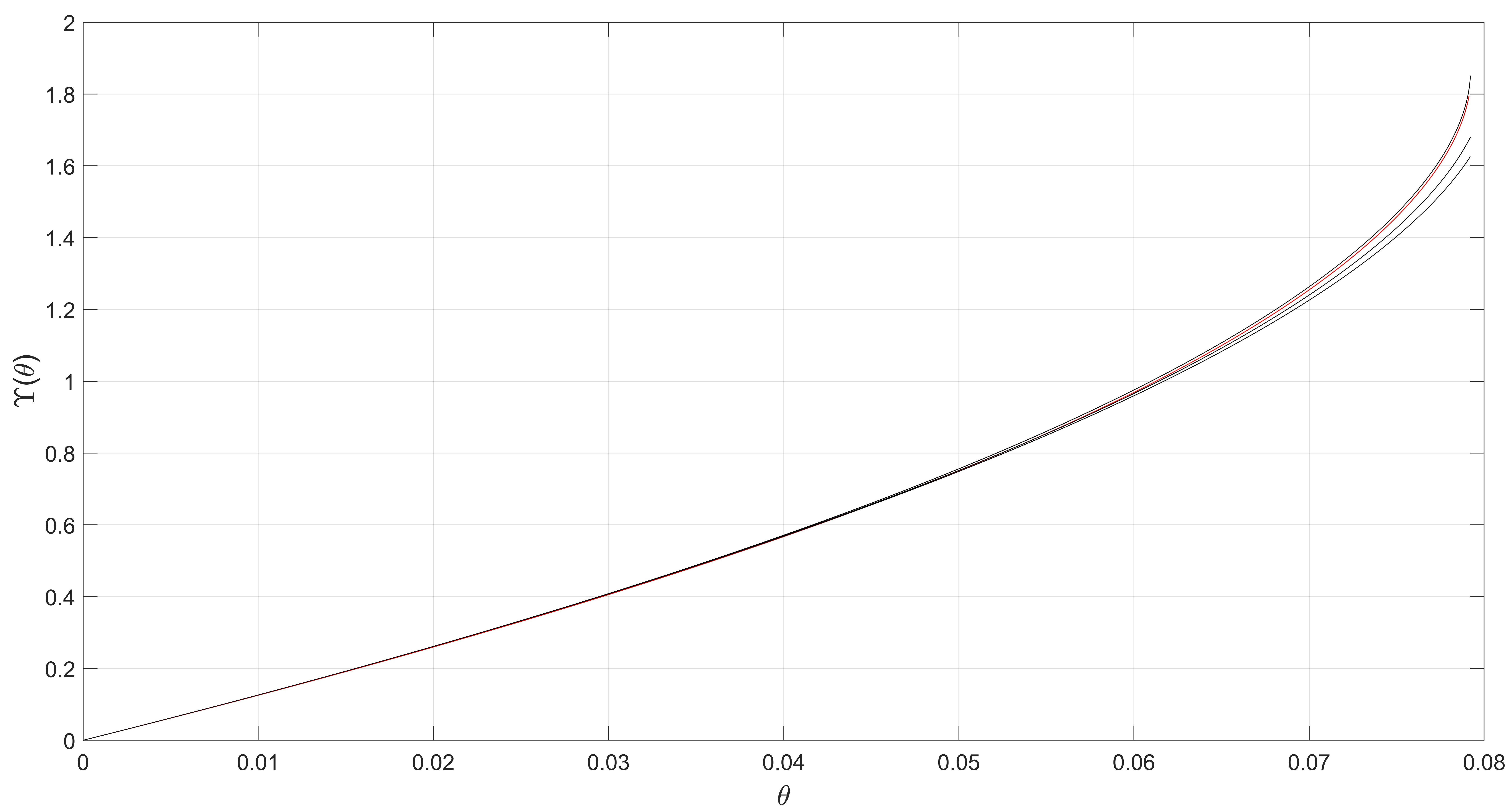}
\caption{The red curve is the same as in Fig.~\ref{fig:Ups1}. The black curves represent the QEF growth rate approximations  $\Ups_r$ from  (\ref{Upslim}) via Theorem~\ref{th:QEFr} for $r=0,1,2,3$ with the alternative coefficients corresponding to the square root polynomial approximation (\ref{phigr}) (the lower curve corresponds to $\Ups_0$). }
\label{fig:Ups2}
\end{center}
\end{figure}
visualises the alternative versions of the QEF growth rate approximations,   which are computed using   the square root polynomial approximation (\ref{phigr})  and appear to be of similar quality.

\section{Conclusion}
\label{sec:conc}

A method has been developed for the state-space  computation of the QEF growth rate for stable open quantum harmonic oscillators driven by vacuum input fields. This has been achieved by relating the frequency-domain representation of the  quantum  QEF rate to a similar functional for a stationary Gaussian random process produced by an infinite cascade of linear systems.
The infinite-dimensional  shaping filter has resulted from a spectral factorization of a special  entire function evaluated at a rational transfer matrix and is found by solving a sequence of ALEs. The latter have been obtained by using a system transposition technique for rearranging  a mixed product of linear systems and their duals similarly to the Wick ordering of annihilation and creation operators. A truncation of the ALE sequence is complemented by an ARE, which allows the QEF rate to be computed with any accuracy. Despite the rapid convergence due to the presence of factorially fast decaying coefficients, a circle of ideas has also been discussed towards a recursive solution of the Riccati equation with respect to the order of truncation and also for a square root polynomial approximation of the entire function.
The state-space computation of the QEF rate can be applied to the large deviations bounds on tail distributions for quantum system variables \cite{VPJ_2018a} and for guaranteed upper bounds on LQG costs in the case of quantum statistical uncertainties with a von Neumann relative entropy description \cite{VPJ_2018b} in robust performance analysis problems. These state-space methods are also applicable to risk-sensitive quantum control problems for OQHOs with QEF optimality criteria \cite{VJP_2020_IFAC}.

\section*{Acknowledgement}

This work is supported by the Australian Research Council grant DP210101938.

\end{document}